%% file: main.tex
\pdfoutput=1

\documentclass[12pt, reqno]{amsart}

\usepackage{main}

\hypersetup{
  pdftitle={How Reliable Are Psychological Measurements? The Distribution of
    Marginal Reliability Across 889 Item-Response Datasets},
  pdfauthor={JoonHo Lee},
}

\graphicspath{{figures/}}
\usepackage{titletoc}

\usepackage[style=apa, backend=biber, natbib=true]{biblatex}
\newgeometry{margin=1.25in}

\title{\Large How Reliable Are Psychological Measurements?\\
The Distribution of Marginal Reliability Across 889 Item-Response Datasets}

\author{JoonHo Lee}

\date{\footnotesize July 2026. \\[0.5em]
Lee: College of Education, The University of Alabama, Tuscaloosa, AL, USA.
\texttt{jlee296@ua.edu}.
}

\renewcommand{\theHequation}{\thesection.\arabic{equation}}
\renewcommand{\theHfigure}{\thesection.\arabic{figure}}
\renewcommand{\theHtable}{\thesection.\arabic{table}}

\begin{document}


\begin{abstract}
Nearly every quantitative study in psychology reports a reliability
coefficient, so the field knows a great deal about the reliability that
authors choose to publish. It knows much less about the reliability of the
data psychology actually produces, because published coefficients pass
through decisions about what to compute and what to report. We therefore
measure reliability directly, applying the same estimators under the same
rules to 889 datasets from the Item Response Warehouse, a public collection
of item-response data that spans cognitive tests, clinical screeners,
personality inventories, and attitude scales. Three findings emerge. Low
reliability is common: even under a lenient definition of reliability,
30\% of datasets fall below the conventional .80 threshold. The variation
across datasets is real rather than statistical, since estimation noise
accounts for only about one percent of it. Finally, the answer depends on
the definition itself: under a strict definition the share below .80 rises
to 52\%, a difference large enough to change what one concludes about the
field. We conclude that a reliability report should say which definition it
uses, attach a measure of uncertainty, and give a strict coefficient
alongside a lenient one.
\end{abstract}

\maketitle


\pagestyle{plain}

\newpage


\input{sections/body}


\bigskip

\noindent\textbf{Acknowledgments.}\enspace
The author thanks the Item Response Warehouse team for building and
maintaining the public data resource on which this study depends.

\smallskip

\noindent\textbf{Software and reproducibility.}\enspace
The analysis is a fully scripted R pipeline (R 4.6.0): item response models
are fitted with \pkg{mirt} 1.46.1 \citep{chalmers_mirt_2012} and \pkg{TAM}
4.3-25 \citep{robitzsch_tam_2026}, person-bootstrap standard errors are
computed on a 30-worker cloud instance with deterministic per-dataset
seeds, analytic standard errors follow \citet{ark_standard_2025} and
\citet{andersson_large_2018} (\pkg{irtreliability};
\citealt{andersson_irtreliability_2022}), and the three-level deconvolution
and meta-regression use \pkg{metafor} 5.0-1
\citep{viechtbauer_conducting_2010} with CR2 robust variance estimation
from \pkg{clubSandwich} 0.7.0 \citep{pustejovsky_clubsandwich_2026}.
Released floats and figures regenerate from the canonical derived tables
with a single build command, and the accompanying artifact manifest records
their input, script, and output SHA-256 hashes. The release also includes
the derived unit-level outputs (point estimates, standard errors, and
moderators) for the datasets it is able to redistribute; see the data
availability statement. This claim concerns the
canonical-derived-to-manuscript build; it does not assert that every
historical raw-data fit is invalidated and rerun by a content-addressed
ledger.

\smallskip

\noindent\textbf{Data availability.}\enspace
All item-response data analyzed here come from the public Item Response
Warehouse \citep[IRW;][]{domingue_introduction_2025},
\url{https://datapages.github.io/irw/}. The analysis uses a frozen local
snapshot (v1.2, acquired April 2026). Of 930 unique eligible administration
units attempted, 889 produced valid primary estimates. The replication
package, \url{https://github.com/joonho112/irw-reliability-replication},
distributes the derived unit-level table underlying every analysis (up to six
reliability point estimates with coefficient-specific coverage, bootstrap and
analytic standard errors, and design moderators) together with the code that
regenerates every number, table, and figure reported here without access to
the item responses. The released table covers 879 of the 889 datasets: ten
sources that the Warehouse marks as not publicly re-shareable are withheld at
the row level, which moves the pooled coefficients by less than .001 and the
deconvolved sub-.80 shares by less than .24 percentage points. Original
response-level data remain governed by the licenses of the primary IRW
sources.


\printbibliography[title={References}]


\newpage
\appendix
\setcounter{section}{0}
\renewcommand{\thesection}{\Alph{section}}
\numberwithin{equation}{section}
\numberwithin{figure}{section}
\numberwithin{table}{section}

\renewcommand{\theHequation}{Supp.\thesection.\arabic{equation}}
\renewcommand{\theHfigure}{Supp.\thesection.\arabic{figure}}
\renewcommand{\theHtable}{Supp.\thesection.\arabic{table}}

\crefname{section}{Appendix}{Appendices}
\Crefname{section}{Appendix}{Appendices}

\begin{center}
{\Large\bfseries Online Supplement}\\[1em]
{\large How Reliable Are Psychological Measurements?\\
The Distribution of Marginal Reliability Across 889 Item-Response Datasets}\\[1.5em]
{\normalsize JoonHo Lee}
\end{center}

\vspace{1.5em}

\noindent This supplement accompanies the main text. Section, equation,
table, and figure numbers below are prefixed by the appendix letter (A.1,
B.1, \dots); references to the main text cite it by section, figure, and
table number.

\vspace{1.5em}

\startcontents[appendices]
\printcontents[appendices]{}{1}{\textbf{Contents}\vskip1em\hrule\vskip1em}
\vskip1em\hrule\vskip2em

\input{appendices/appendix_A}
\input{appendices/appendix_B}
\input{appendices/appendix_C}
\input{appendices/appendix_D}
\input{appendices/appendix_E}
\input{appendices/appendix_F}
\input{appendices/appendix_G}

\end{document}

%% file: sections/body.tex

\section{Introduction}\label{sec:intro}

In psychology, nearly every quantitative study reports a reliability
coefficient, and the reporting follows a settled routine. A coefficient is
computed on the study's own sample and compared with a conventional
threshold of .70, .80, or .90, usually without a standard error
\citep{parsons_psychological_2019}. These thresholds are commonly attributed
to \citet{nunnally_psychometric_1994}, although what Nunnally recommended
was .70 for early-stage research and .90 for applied decisions
\citep{lance_sources_2006, greco_meta-analysis_2018}. Treating measurement
precision so casually has consequences that are by now well documented.
\Citet{loken_measurement_2017} showed that measurement error, combined with
selection on significance, can reverse inferential conclusions and not
merely attenuate them. \Citet{hedge_reliability_2018} and
\citet{rouder_hierarchical_2024} showed that robust experimental effects can
coexist with unreliable individual-difference measures, and
\citet{flake_measurement_2020} documented how questionable measurement
practices propagate through substantive literatures.

Much is known about the distribution of \emph{reported} reliability
coefficients. \Citet{greco_meta-analysis_2018} cumulated coefficient alpha
across 1{,}675 independent samples and concluded that reported values almost
always exceed .70 and generally exceed .80. \Citet{hussey_aberrant_2025}
extracted more than 30{,}000 alpha values from APA-journal articles, a
further 89{,}000 from the industrial--organizational literature and 67{,}000
from the PsycTests database, and found excess mass precisely at the
conventional thresholds, which indicates that the values reaching print have
passed through a filter. Since \citet{vachahaase_reliability_1998}, the
reliability generalization literature has meta-analyzed reported coefficients
one instrument at a time \citep{thompson_psychometrics_2000,
sanchezmeca_recommended_2013}. Its methodological branch has established that
reliability varies materially across administrations of the same instrument,
and also that the analytic choices made in synthesizing coefficients affect
the conclusions \citep{lopez-ibanez_reliability_2024}.
\Citet{scherer_tutorial_2020} went a step beyond harvesting. By
meta-analyzing correlation matrices instead of reported coefficients, they
could estimate the same coefficient uniformly across samples, and for a
single instrument, the Rosenberg Self-Esteem Scale, they obtained model-based
reliabilities ranging from .41 to .86 across 34 samples. On one point these
literatures agree: reliability is a property of an administration and not of
an instrument.

What none of them can supply is the distribution of reliability as
\emph{computed}, that is, the coefficients obtained when the same estimators
are applied under the same rules to every dataset in a defined collection,
whether or not anyone chose to compute or publish them. This is the quantity
a reader wants when a paper reports no coefficient at all, and it is the
natural baseline against which the optimism of the reported distribution can
be assessed. Obtaining it requires raw item responses at scale. Until
recently the response matrices were scattered, heterogeneous, and mostly
private.

The Item Response Warehouse \citep[IRW;][]{domingue_introduction_2025} has
changed this. It harmonizes hundreds of publicly available item-response
datasets into a common long format, covering cognitive tests, clinical
screeners, personality inventories, attitude scales, and behavioral tasks. A
research program has already formed around it. \Citet{zhang_realistic_2025}
characterized item-difficulty distributions across 73 datasets,
\citet{nalbandyan_signposts_2026} studied the nominal-to-ordinal transition
across 245, and \citet{domingue_intermodel_2024} compared item response
models across 89. More recently, \citet{liu_comparing_2026} compared
multidimensional model specifications across 25 large datasets, and
\citet{gilbert_estimating_2025} and \citet{gilbert_conditional_2026}
conducted item-level meta-analyses of treatment effects and of response-time
dependencies. Reliability has not yet been examined in this way. In this
paper we attempt every one of the 930 eligible datasets in a frozen snapshot
of the warehouse; 889 produce valid estimates, and for each we compute up to
six reliability coefficients together with their standard errors.

Before any such computation can be interpreted, a difficulty has to be faced.
What is called marginal reliability is not a single estimand but a family of
related estimands, distinguished by which error variance is averaged and over
which distribution of the latent trait. The members are not interchangeable.
We prove a total order among three information-based latent-scale
functionals, so that two defensible answers to the question of how reliable
a set of scores is can differ systematically, and by amounts that matter in
practice. An audit reporting one coefficient per dataset would therefore have
made a choice without saying so. We instead declare two co-primary
operational paths, chosen to contrast a strict separation coefficient with a
lenient posterior coefficient, add four auxiliary and anchor coefficients,
and treat the disagreement among the six as an object of study.

Three results follow from the corpus. First, low reliability is common.
Depending on the estimand, between three datasets in ten and one in two fall
below the conventional .80 threshold, and under the strict estimand roughly
one dataset in nine falls below .50. Second, the dispersion is genuine and
not statistical. Because every point estimate carries a bootstrap standard
error, a three-level random-effects model can separate estimation noise from
true heterogeneity, and the noise accounts for about one percent of the
observed variance. Third, the answer depends materially on the estimator.
Six defensible estimators of the same corpus give pooled reliabilities from
.794 to .868 and sub-.80 shares from 26\% to 52\%, and the disagreement
between the two co-primary estimands is itself structured, being largest for
personality scales.

Two features of the design bear on how the conclusions should be read.
Eligibility rules and estimator definitions were fixed before the corpus
run. The item-key alignment rule is the explicit exception: it was introduced
after a keying defect was discovered during the analysis and then frozen
before the correction rerun. This distinction follows the recommendation that
analytic decisions on preexisting data be settled before results are known
\citep{weston_recommendations_2019}. Our analysis is therefore a census of
the eligible warehouse and not a curated sample. In addition, the corpus was
completed in two waves, with the largest datasets arriving last, so the
composition of the headline can be checked directly. Completing the corpus
in fact made the headline shares slightly worse rather than better, since
the large-sample wave consists mainly of internet panels and surveys rather
than professionally constructed tests (Section~\ref{sec:composition}).

The plan of the paper is as follows. In Section~\ref{sec:family} we develop
the family of estimands, prove the ordering, and declare the estimator
portfolio. Section~\ref{sec:corpus} describes the corpus. In
Section~\ref{sec:methods} we present the estimation pipeline, the
standard-error strategy and its cross-validation, and the deconvolution and
meta-regression models. Results are reported in Section~\ref{sec:results}.
Finally, in Section~\ref{sec:discussion} we discuss implications for
reporting practice, the relation to the literature on reported coefficients,
and limitations. An online supplement (Appendices
\ref{app:family}--\ref{app:additional}) contains proofs,
corpus-construction detail, full tables, and sensitivity analyses.

\section{Estimands for marginal reliability}\label{sec:family}

Every version of marginal reliability translates the classical variance
ratio, true-score variance over total variance \citep{lord_novick_1968},
into item response theory (IRT):
\begin{equation}\label{eq:skeleton}
\rho \;=\; \frac{\sigma^2_\theta}{\sigma^2_\theta + \text{(error summary)}},
\end{equation}
or equivalently $\rho = 1 - (\text{error summary})/(\text{total summary})$.
The versions differ in two choices. The first is which error variance is
summarized: the conditional sampling variance $1/\Jinfo(\theta)$ of a
maximum-likelihood-type score, the posterior variance of a Bayesian score, or
the error variance of an observed sum score. The second is the distribution
over which the summary is taken, either an assumed latent density $g$, in
practice almost always $N(0,1)$, or the realized sample. In
\cref{fig:family} we map the resulting family and locate every estimator
used in this paper.

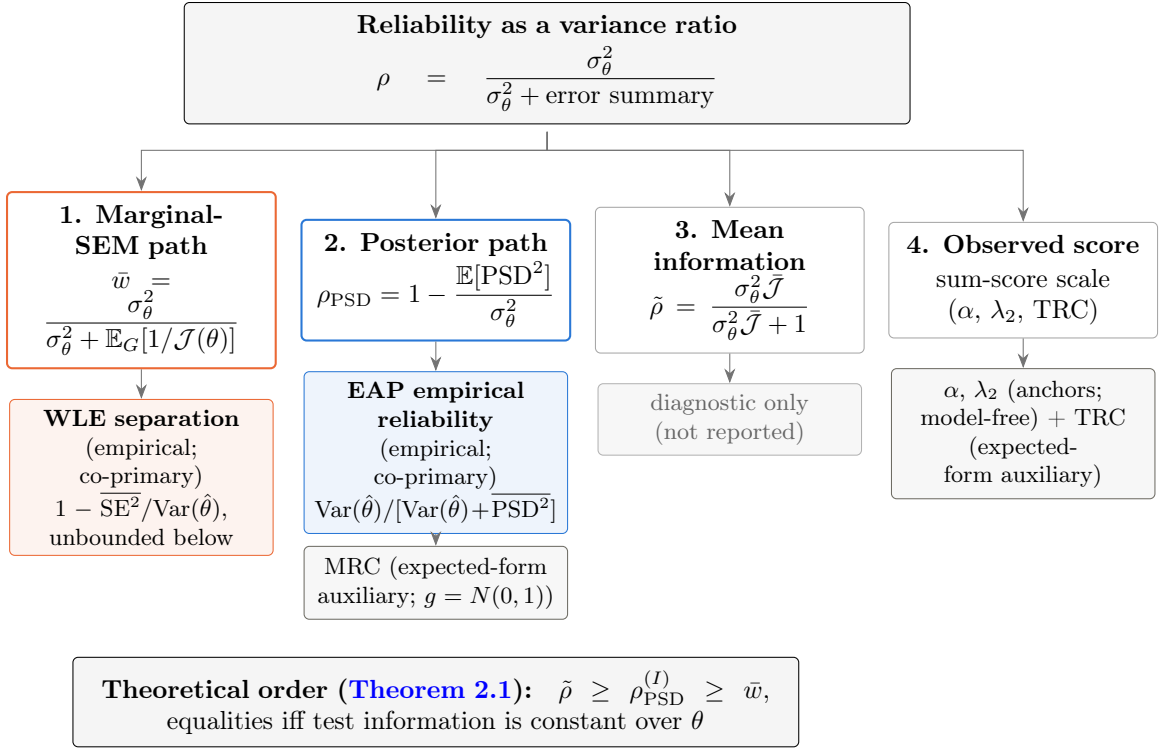
\begin{figure}[t]
\centering
\begin{adjustbox}{max width=\textwidth}
\begin{tikzpicture}[
  font=\footnotesize,
  box/.style={draw, rounded corners=2pt, align=center, inner sep=5pt,
              minimum height=2.4em, text width=0.215\textwidth,
              execute at begin node={\hyphenpenalty=10000\relax}},
  est/.style={draw, rounded corners=2pt, align=center, inner sep=4pt,
              text width=0.215\textwidth, font=\scriptsize,
              execute at begin node={\hyphenpenalty=10000\relax}},
  lab/.style={font=\scriptsize\itshape, text=black!60},
  arr/.style={-{Stealth[length=2mm]}, black!55, thin},
  node distance=4mm and 3.2mm]

\definecolor{eapblue}{RGB}{42,120,214}
\definecolor{wleorange}{RGB}{235,104,52}
\definecolor{auxgray}{RGB}{110,108,100}

\node[box, text width=0.62\textwidth, fill=black!4] (root)
  {\textbf{Reliability as a variance ratio}\\[2pt]
   $\rho \,=\, \dfrac{\sigma^2_\theta}{\sigma^2_\theta + \text{error summary}}$};

\node[box, below=9mm of root.south west, anchor=north west, xshift=-24mm,
      draw=wleorange, thick] (b1)
  {\textbf{1. Marginal-SEM path}\\[1pt]
   $\wbar = \dfrac{\sigma^2_\theta}{\sigma^2_\theta + \E_G[1/\Jinfo(\theta)]}$};
\node[box, right=of b1, draw=eapblue, thick] (b2)
  {\textbf{2. Posterior path}\\[1pt]
   $\rhopsd = 1 - \dfrac{\E[\PSD^2]}{\sigma^2_\theta}$};
\node[box, right=of b2, draw=black!35] (b3)
  {\textbf{3. Mean information}\\[1pt]
   $\rhotilde = \dfrac{\sigma^2_\theta \bar{\Jinfo}}{\sigma^2_\theta \bar{\Jinfo} + 1}$};
\node[box, right=of b3, draw=black!35] (b4)
  {\textbf{4. Observed score}\\[1pt]
   sum-score scale\\ ($\alpha$, $\lambda_2$, TRC)};

\draw[arr] (root.south) ++(0,-1mm) -- ++(0,-2.5mm) -| (b1.north);
\draw[arr] (root.south) ++(0,-1mm) -- ++(0,-2.5mm) -| (b2.north);
\draw[arr] (root.south) ++(0,-1mm) -- ++(0,-2.5mm) -| (b3.north);
\draw[arr] (root.south) ++(0,-1mm) -- ++(0,-2.5mm) -| (b4.north);

\node[est, below=of b1, draw=wleorange, fill=wleorange!8] (e1)
  {\textbf{WLE separation}\\ (empirical; co-primary)\\
   $1 - \overline{\mathrm{SE}^2}/\Var(\hat\theta)$, unbounded below};
\node[est, below=of b2, draw=eapblue, fill=eapblue!8] (e2)
  {\textbf{EAP empirical reliability}\\ (empirical; co-primary)\\
   $\Var(\hat\theta)/[\Var(\hat\theta)+\overline{\PSD^2}]$};
\node[est, below=of b3, draw=black!30, fill=black!3, text=black!60] (e3)
  {diagnostic only\\ (not reported)};
\node[est, below=of b4, draw=auxgray, fill=black!3] (e4)
  {$\alpha$, $\lambda_2$ (anchors;\\ model-free) + TRC\\ (expected-form auxiliary)};

\node[est, below=1.6mm of e2, draw=auxgray, fill=black!3] (e2b)
  {MRC (expected-form\\ auxiliary; $g = N(0,1)$)};

\draw[arr] (b1) -- (e1);
\draw[arr] (b2) -- (e2);
\draw[arr] (b3) -- (e3);
\draw[arr] (b4) -- (e4);
\draw[arr] (e2) -- (e2b);

\node[draw, rounded corners=2pt, fill=black!4, below=5mm of e2b,
      text width=0.62\textwidth, align=center, inner sep=5pt] (order)
  {\textbf{Theoretical order (\cref{thm:order}):}\quad
   $\rhotilde \;\ge\; \rho_{\mathrm{PSD}}^{(I)} \;\ge\; \wbar$,\quad equalities iff test
   information is constant over $\theta$};

\end{tikzpicture}
\end{adjustbox}
\caption{The family of marginal reliability estimands. The classical variance
ratio of \cref{eq:skeleton} branches into four estimands according to which
error variance is averaged. A second choice, between an assumed density $g$
and the realized sample, splits each latent-scale branch into expected and
empirical versions. The boxes below each branch show the estimators used in
this paper: two empirical co-primaries (WLE separation, orange; EAP empirical
reliability, blue), two expected-form auxiliaries (MRC, TRC), and two
model-free classical-test-theory anchors ($\alpha$, $\lambda_2$). The stated
order concerns the three information functionals of
Appendix~\ref{app:family}.}
\label{fig:family}
\end{figure}

\subsection{Four summaries of error variance}\label{sec:branches}

The first branch averages the error variance of maximum-likelihood-type
scores, the marginal standard error of measurement, giving
\begin{equation}\label{eq:wbar}
\wbar \;=\; \frac{\sigma^2_\theta}{\sigma^2_\theta +
\E_G\!\left[1/\Jinfo(\theta)\right]}.
\end{equation}
This is the lineage of \citet{green_technical_1984}, of
expected-inverse-information summaries \citep{cheng_comparison_2012,
kim_note_2012}, and of person separation reliability in the Rasch tradition
\citep{wright_masters_1982, adams_reliability_2005}. Since the integrand
$1/\Jinfo$ is unbounded, regions of $\theta$ where the test carries little
information can dominate the average, and finite-sample estimates can fall
below zero. This behavior is not a defect of the estimand. The quantity
$\wbar$ answers the signal-to-noise question faced by a user who scores
respondents with unshrunken, maximum-likelihood-type estimates
\citep{cronbach_signalnoise_1964}.

The second branch averages posterior variances of Bayesian scores, giving
\begin{equation}\label{eq:rhopsd}
\rhopsd \;=\; 1 - \frac{\E[\PSD^2]}{\sigma^2_\theta} \;\approx\; 1 -
\frac{1}{\sigma^2_\theta}\,\E_G\!\left[\frac{1}{\Jinfo(\theta) +
1/\sigma^2_\theta}\right],
\end{equation}
the reliability of expected a posteriori (EAP) scores
\citep{bock_adaptive_1982}. Here the prior variance truncates the integrand,
so that $\rhopsd$ is finite for any $g$ and cannot fall below zero. Shrinkage
thus yields a more stable quantity, and also a more lenient one. The second
expression is the standard Gaussian approximation; the formal results below
concern this information form, written $\rho_{\mathrm{PSD}}^{(I)}$.

The third branch averages the information function first,
$\bar{\Jinfo} = \E_G[\Jinfo(\theta)]$, and then transforms, yielding
$\rhotilde = \sigma^2_\theta\bar{\Jinfo}/(\sigma^2_\theta\bar{\Jinfo}+1)$
\citep{adams_reliability_2005}. Whether averaged information yields a
reliability at all has been debated \citep{doran_information_2005}. We treat
this design-effect-style summary as a diagnostic and do not report it
\citep[cf.][]{lee_reliability_2025}.

The fourth branch is the reliability of the unweighted sum score. It contains
the model-free coefficients $\alpha$ and $\lambda_2$
\citep{guttman_basis_1945, cronbach_coefficient_1951} together with the
model-implied test reliability coefficient \citep[TRC;][]{kim_estimation_2010,
andersson_large_2018}. We omit $\omega$ and the greatest lower bound. A
common argument is that they add structural commitments orthogonal to an
audit \citep{mcdonald_test_1999, revelle_reliability_2019}. A sharper one was
supplied by \citet{berge_greatest_2004}, who showed that
unidimensionality-based coefficients rest on a hypothesis that cannot be
tested with three test parts and is false with more, and that for more than
three parts the greatest lower bound and $\omega$ behave almost identically
in any case, including in their small-sample bias.

Branch 4 is defined on a different score scale from branches 1--3, and we use
it as an anchor rather than as a competitor. Two of its properties matter for
interpretation. First, $\alpha$ and $\lambda_2$ are lower bounds on
sum-score reliability, so their distance below it reflects the violation of
tau-equivalence and not a different estimand. In three simulation studies,
\citet{kelley_confidence_2016} found that no confidence-interval method for
$\alpha$ bracketed population reliability acceptably, and they recommended
against using $\alpha$ intervals for that purpose. Second, a lower bound is
still informative. \Citet{sijtsma_part_2021} argue that lower bounds function
as quality guarantees and that for approximately unidimensional data the
discrepancy is small. A long debate on $\alpha$'s bias and its alternatives
\citep{sijtsma_use_2009, mcneish_thanks_2018, raykov_importance_2023}
concerns this branch alone. Cronbach himself came to regard the coefficient
as a restricted special case within a larger variance-components system
\citep{cronbach_my_2004}. Note also that which scale is the more reliable,
the sum score or the latent estimate, is not fixed.
\Citet{culpepper_reliability_2013} showed analytically that the ordering
depends on the match between item locations and the latent distribution, the
sum score being the more reliable when items sit in the tails of the trait
distribution. Strictly speaking, the term marginal, in the sense of
population-averaged over $G$, applies to branches 1--3, and the
observed-score branch enters as a model-free anchor.

\subsection{Expected and empirical versions}\label{sec:versions}

Orthogonal to the choice of branch, each latent-scale estimand has an
expected version, integrated over an assumed $g$ that is almost always
$N(0,1)$, and an empirical version, averaged over the realized sample of
scores and their standard errors \citep{zhang_empirical_2006}. Expected
versions are sample-free, but they inherit the normality assumption, which is
itself open to question in an audit of this kind. Empirical versions are
assumption-light but sample-dependent, and under restricted or adaptively
selected samples they can behave dramatically, including negative separation
reliabilities. We therefore take the empirical versions as primary and the
expected versions, MRC for branch 2 and TRC for branch 4, as auxiliaries.

\subsection{The item facet}\label{sec:itemfacet}

Both choices so far hold the item set fixed, so that every coefficient in
\cref{fig:family} describes how precisely these particular items order
respondents. Generalizability theory asks a further question
\citep{brennan_generalizability_2001}. If items are regarded as a sample from
a content domain, a user who cares about the construct rather than the form
wants precision generalized over that domain as well. Treating items as a
random facet adds an item-by-person interaction component to the error term,
so that the generalizability coefficient is bounded above by its item-fixed
counterpart, with equality only when items are exchangeable replicates.

In practice the gap is not negligible. Working with item-level data from
eight studies, \citet{gilbert_item-level_2025} found that ignoring
cluster-by-item interaction overstates reliability by a median of .04 (mean
.09) on the reliability scale in a value-added setting, and item-level
reanalyses of randomized trials find that treating items as a sample rather
than as fixed adds uncertainty large enough to change conclusions
\citep{ahmed_heterogeneity_2025}. All six estimators in this paper are
item-fixed, and we do not estimate the item-random branch. Identifying it
requires alternate forms or a content model specifying the domain, neither of
which the warehouse supplies. The cited studies establish a direction of
concern, namely that conditioning on the realized items can omit
item-sampling uncertainty, but their effect sizes do not turn the present
sub-threshold shares into formal lower bounds. A reader interested in
generalizability over items should therefore treat this as an unresolved,
plausibly optimistic limitation. A second optimistic layer is discussed in
Section~\ref{sec:points}.

\subsection{An ordering of the latent-scale estimands}\label{sec:order}

The three latent-scale branches are not interchangeable, and their
disagreement has a sign.

\begin{theorem}\label{thm:order}
Let $0 < \sigma^2_\theta < \infty$, let $\Jinfo(\theta)>0$ $G$-almost surely,
and suppose $\E_G[\Jinfo]$ and $\E_G[1/\Jinfo]$ are finite. For the three
information functionals in Appendix~\ref{app:family},
\[
\rhotilde \;\ge\; \rho_{\mathrm{PSD}}^{(I)} \;\ge\; \wbar,
\]
with equalities if and only if $\Jinfo(\theta)$ is $G$-almost surely constant.
Under constant information $\Jinfo \equiv I$, all three reduce to
$\rho = I\sigma^2_\theta / (I\sigma^2_\theta + 1)$.
\end{theorem}

The proof, given in Appendix~\ref{app:family}, applies Jensen's inequality
twice. The map $u \mapsto \sigma^2/(\sigma^2+u)$ is convex in the error
summary while $x \mapsto 1/(x + c)$ is convex in information, and the two
inequalities enclose $\rho_{\mathrm{PSD}}^{(I)}$ between the other branches.
A numerical audit over $2\times 10^5$ random test designs found no
violations, and six worked scenarios, including a negative-$\wbar$ case
constructed to mirror a real dataset, are reported in
Appendix~\ref{app:family}. In practice the information curve is never flat,
so for essentially every real test the information-form EAP-type reliability
exceeds its WLE-type counterpart. That posterior-based scores attain higher
reliability than maximum-likelihood scores has been observed since early
comparisons of ability estimators \citep{kim_nicewander_1993}. What the
theorem adds is the two-sided ordering, its equality condition, and the
placement of the separation-type quantity at the bottom. Whether the ordering
carries over to the corpus estimators, which use two calibrations and
realized score variances, is an empirical question rather than a corollary.
It does: 97.1\% of units order as predicted, with 25 reversals of median
magnitude .002 (Section~\ref{sec:distribution}). An audit that reports a
single marginal reliability has therefore made a choice among branches,
whether or not it says so.

\subsection{The estimator portfolio}\label{sec:portfolio}

Guided by the family, we declare two co-primary operational estimands that
contrast common strict and lenient scoring paths; they are not universal
bounds on every reliability branch. The first is WLE separation reliability,
on branch 1 in its empirical version. It uses one calibration per dataset,
with scores and standard errors obtained by weighted likelihood estimation
\citep{warm_weighted_1989} and reliability computed as
$1 - \overline{\mathrm{SE}^2}/\Var(\hat\theta)$. Since it is unbounded below,
negative values are retained as information and are not clipped. The second
is EAP empirical reliability, on branch 2 in its empirical version. It uses
EAP scores and posterior standard deviations from the same calibration
family, with reliability computed as
$\Var(\hat\theta)/[\Var(\hat\theta) + \overline{\PSD^2}]$, the sample analog
of \cref{eq:rhopsd} \citep{bock_adaptive_1982, chalmers_mirt_2012}.

Around the co-primaries we report four further coefficients. Two are the
expected-form auxiliaries MRC and TRC \citep{andersson_large_2018}, which
import $g = N(0,1)$, and two are the model-free anchors $\alpha$ and
$\lambda_2$, computable without any IRT model. Estimating six coefficients on
one corpus turns estimator choice from a hidden researcher degree of freedom
into a measured quantity, in the spirit of multiverse and
specification-curve reporting \citep{steegen_increasing_2016,
simonsohn_specification_2020}.

\section{The corpus}\label{sec:corpus}

\input{floats/tab1_corpus}

Our data source is a frozen local snapshot (v1.2, acquired April 2026) of
the Item Response Warehouse \citep{domingue_introduction_2025}. Our unit of
analysis is the \emph{administration unit}, one response matrix from one
table--wave combination. Longitudinal tables contribute their most complete
wave, selected by a documented rule, because reliability is a property of an
administration and not of a panel. An indicator for repeated-measures origin
enters the meta-regression as a moderator (Appendix~\ref{app:corpus}).

To be eligible, a dataset needed at least three items, at least two response
categories, and at least 100 respondents. The warehouse inventory yields 933
eligible records; three datasets appear twice, because the snapshot combines
a canonical collection with a provisionally ingested one, and we keep one
copy of each, leaving 930 unique datasets. Two pre-declared caps then remove
formats outside the estimand's scope: 23 datasets whose items have more than
15 response categories, which are slider and visual-analog formats, and 17
datasets with more than 600 items, which are large and mostly adaptive item
banks. One further dataset was degenerate. Each exclusion carries a single
machine-readable reason code, so datasets satisfying more than one rule are
counted once. All 930 datasets were attempted, and 889 (95.6\%) produced
valid estimates with no unexplained failures (\cref{tab:corpus}).

Processing itself involved two steps. Item responses were densely recoded to
$0, 1, \dots$ per item, and item keys were aligned, meaning that items
loading against a scale's dominant direction were reverse-scored so that
sum-score and latent-scale coefficients refer to the same construct
orientation. The alignment rule was introduced after a keying defect was
discovered during the analysis, was frozen before the correction rerun, and
is documented with its validation in Appendix~\ref{app:keyalign}. Finally,
the corpus was computed in two waves: a first wave of 788 smaller datasets,
and a second wave of the 101 largest ones. Within the second wave, every
dataset with more than 5{,}000 respondents was analyzed on a fixed random
sample of 5{,}000 persons drawn with a deterministic per-dataset seed; this
applies to 89 of the 101. Reliability is a population functional, so
subsampling persons affects only estimation precision, and the sampling
error this induces, about .003--.006, is an order of magnitude below the
between-dataset heterogeneity of interest (Appendix~\ref{app:corpus}).

The resulting corpus spans 440 study clusters and contains 173 cognitive or
educational units, 201 affective or mental-health units, 178 opinion or
attitude units, and 166 personality, behavioral, developmental,
physical-health, or other units, with 171 unclassified. In the median unit
there are 909 respondents and 18 items, and 231 units are dichotomous
against 658 polytomous. As far as we are aware, no reliability computation
of comparable breadth exists under a common pipeline. Because eligibility is
rule-based over a public warehouse, the corpus approximates the datasets
psychology produces and shares, subject to the selection frame discussed in
Section~\ref{sec:limitations}.

\section{Estimation and inference}\label{sec:methods}

\subsection{Point estimation}\label{sec:points}

Each unit receives one primary calibration by marginal maximum likelihood in
\pkg{mirt} \citep{chalmers_mirt_2012}: the Rasch model for dichotomous units
\citep{rasch_probabilistic_1960} and the generalized partial credit model
\citep[GPCM;][]{muraki_generalized_1992, masters_partial_1982} for polytomous
units. From this calibration we obtain the EAP empirical reliability directly.
WLE separation reliability comes from an independent calibration of the same
matrix in \pkg{TAM} \citep{robitzsch_tam_2026}, whose weighted-likelihood
scoring implements \citet{warm_weighted_1989}.

The two co-primaries therefore rest on different calibrations and, for
polytomous units, on different models, since \pkg{TAM}'s default polytomous
specification is the partial credit model with a common slope whereas
\pkg{mirt}'s GPCM estimates a free slope per item. This follows from what the
two estimands are. Separation reliability belongs to the Rasch tradition,
which works in fixed-slope models, while posterior-based reliability is
routinely computed under free-slope calibrations, so the pair compares the
two quantities as they are actually produced. For interpretation the
consequence is that on polytomous units the WLE--EAP contrast mixes an
estimand difference with a model difference. Dichotomous units, where both
packages fit the same Rasch model, isolate the estimand contrast: their EAP
reliabilities agree across packages to a median absolute difference of .0004
($n = 224$), whereas on polytomous units the packages differ by a median of
.029 ($n = 655$, with 73\% exceeding .01). In Appendix~\ref{app:estimation}
we report a sensitivity analysis in which \pkg{TAM} is refitted with free
slopes on 30 stratified polytomous units. There the median change in WLE
separation reliability is $+.009$, but 23\% of units move by more than .05,
so the model choice is a real source of unit-level variation even though it
does not shift the corpus-level answer. For the anchors $\alpha$ and
$\lambda_2$ we use complete cases and report a coefficient when item-level
missingness is below 5\% and at least 100 complete cases remain. MRC and TRC
require the item-parameter covariance matrix and are computed for units with
at most 100 items \citep{andersson_irtreliability_2022}. Coverage by
estimator appears in \cref{tab:portfolio}, and per-unit seeds, convergence
settings, and the full processing record are given in
Appendix~\ref{app:estimation}.

One property of scoring from a fitted calibration should be recorded here.
Posterior standard deviations, and hence EAP empirical reliability, condition
on the estimated item parameters, and ignoring that layer of uncertainty is
known to overstate the reliability of IRT scale scores
\citep{cho_multilevel_2019}. Our bootstrap propagates the omitted layer into
the EAP standard errors, since each replicate re-estimates the full model
(Section~\ref{sec:ses}), but it does not remove the bias from the point
estimate, which therefore remains slightly upward. Like the item-facet gap of
Section~\ref{sec:itemfacet}, this bias runs against the headline, so that the
true shares below threshold are, if anything, larger than we report.

\subsection{Standard errors}\label{sec:ses}

Not every coefficient carries a standard error, so one accounting rule is
used throughout the paper: pooled analyses use, for each estimator, the
subset of units with a usable standard error. \Cref{tab:portfolio} reports
that count for every coefficient; under the co-primaries it is 877 for EAP
and 867 for WLE of the 889 units. The pooled models apply one further
restriction, excluding as numerically degenerate the units whose separation
reliability falls below $-1$. Two of those carry a WLE standard error, so
the WLE deconvolution and meta-regression are fitted on 865 units, which is
the count reported with those models. Standard errors come from a hybrid strategy, with the
person bootstrap for the co-primaries. For EAP reliability each replicate
re-estimates the full model. The cost-aware rule has an attempted floor of
12 replicates, but failed refits reduce the successful count: its median is
100, its lower quartile 40, 19.9\% of units have fewer than 30 successful
replicates, and 11 have at most three and therefore no EAP standard error.
In Appendix~\ref{app:ses} we show that the resulting Monte Carlo error is
small after aggregation. For WLE separation reliability we hold the
calibration fixed and resample the realized $(\hat\theta_i, \mathrm{SE}_i)$
pairs, because refitting \pkg{TAM} on every draw is computationally
prohibitive at corpus scale. Our reported WLE uncertainty is therefore
conditional on the calibration \citep[on such uncertainty layers,
see][]{yang_characterizing_2012}, and we measured the shortfall directly.
Full-refit \pkg{TAM} bootstraps were attempted on 30 cross-validation units;
29 yield usable paired estimates and give a median full-to-pair
standard-error ratio of 1.06 (interquartile range 0.95--1.20), and
propagating the implied variance factor of 1.12 through the deconvolution
changes no true-share estimate by more than 0.1 percentage points
(Appendix~\ref{app:ses}). For the anchors we use the analytic
multinomial-delta standard errors of \citet{ark_standard_2025}, and for MRC
and TRC the item-parameter delta method of \citet{andersson_large_2018}.
Median standard errors are .006 for EAP and .010 for WLE on the reliability
scale.

Since the deconvolution below treats these standard errors as known, we
cross-validated the analytic ones against the person bootstrap on a
cheap-refit sample of 30 attempted units, with 29 usable for several
summaries (Appendix~\ref{app:ses}). This design underrepresents large units
and excludes slow refits by construction. The van der Ark delta method is
essentially exact, with median bootstrap-to-analytic ratios of 1.006 for
$\alpha$ and 1.003 for $\lambda_2$ and log-standard-error correlations of
.998. By contrast, the item-parameter delta method is a mild lower bound,
with median ratios of 1.02 for MRC and 1.15 for TRC. Note that this gap does
not shrink with sample size, which identifies its source as uncertainty the
delta method conditions away, namely model-specification propagation and the
fixed $N(0,1)$ prior, rather than as small-sample bias. All four medians
fall inside the pre-specified acceptance band $[0.80, 1.25]$, but TRC clears
it only on the median: 38\% of its units sit in band, and its
log-standard-error correlation with the bootstrap is .64, the weakest
agreement in the portfolio. We therefore regard TRC's standard error as the
least trustworthy in the portfolio, use bootstrap standard errors wherever
both exist, and treat the delta-method auxiliaries as mildly
anticonservative.

\subsection{The distribution of true reliability}\label{sec:deconv}

Observed coefficients scatter around true values, so the observed
distribution overstates the dispersion of the true one. To recover the
distribution of true reliability we work on the transformed scale
$L = \ln(1-r)$, which removes the ceiling at $r = 1$, accommodates negative
WLE values, and stabilizes variances \citep{bonett_sample_2002,
bonett_varying_2010}. Comparisons of transformations in the reliability
generalization literature support this choice
\citep{lopez-ibanez_reliability_2024}. Delta-method standard errors transform
as $\mathrm{SE}_L = \mathrm{SE}_r/(1-r)$. For unit $i$ in study $s$ we fit
\begin{equation}\label{eq:decon}
\hat L_{si} \;=\; \mu + u_s + u_{si} + \varepsilon_{si}, \qquad u_s \sim
N(0, \tau^2_{\mathrm{study}}), \quad u_{si} \sim N(0, \tau^2_{\mathrm{unit}}),
\quad \varepsilon_{si} \sim N(0, v_{si}),
\end{equation}
a three-level random-effects model \citep{konstantopoulos_fixed_2011,
vandennoortgate_threelevel_2013} estimated by restricted maximum likelihood
in \pkg{metafor} \citep{viechtbauer_conducting_2010}. The estimated bootstrap
variances $v_{si}$ are treated as known, and the within-study sampling
correlation is handled by the same working model as the meta-regression
below ($\rho = 0.6$; \citealt{pustejovsky_meta-analysis_2022}), so that
every model in the paper makes one dependence assumption rather than two
(Appendix~\ref{app:deconv-rho}). Three-level specifications of this kind
have precedent in the meta-analysis of model-based reliability coefficients
\citep{scherer_tutorial_2020}. Because $v_{si}$ is supplied rather than
estimated jointly, restricted maximum likelihood can separate the modeled
sampling-noise layer from heterogeneity. This is the classical
known-variance random-effects decomposition, which acts as a parametric
empirical-Bayes deconvolution \citep[cf.][]{efron_deconvolution_2016}. Under
this model the implied true share below a threshold $c$ is
$1 - \Phi\{[\ln(1-c) - \hat\mu]/\hat\tau_{\mathrm{total}}\}$, and the true
density back-transforms in closed form (Appendix~\ref{app:deconv}).

The population to which these parameters refer requires comment. In the
terms of \citet{bonett_varying_2010}, a random-effects model is a
random-coefficient model whose parameters refer to a superpopulation, and
our corpus was not sampled from one but is a census of the eligible
warehouse. We therefore also report the share below threshold among these
889 units computed from Gaussian empirical-Bayes unit predictions. This is a
diagnostic, not Bonett's varying-coefficient estimator: it avoids a
normal-tail formula but retains the fitted Gaussian hierarchy and shrinkage.
At the .80 threshold it is within about two percentage points of the normal
tail (Appendix~\ref{app:deconv-target}). Where a superpopulation reading is
intended, it should be understood as instruments of the kind that reach a
public warehouse.

\subsection{Meta-regression}\label{sec:metareg-m}

To ask which design features carry the heterogeneity, we extend
\cref{eq:decon} with moderators: the logarithm of the number of items, the
logarithm of the number of respondents (log items and log persons below),
the number of response categories, construct type, and the
repeated-measures indicator. This set is deliberately confirmatory and was
fixed in advance. Quantities that are mechanically posterior to the
calibration, such as a unidimensionality index computed from the same fit,
are excluded as mediators and examined separately
(Appendix~\ref{app:metareg}). Units within a study often share respondents,
for instance subscales of one battery, so their sampling errors are
correlated with unknown correlation. We impute a working $\rho = 0.6$ in the
correlated-and-hierarchical-effects structure of
\citet{pustejovsky_meta-analysis_2022} and base all inference on robust
variance estimation \citep{hedges_robust_2010} with CR2 cluster-robust
standard errors and Satterthwaite degrees of freedom
\citep{satterthwaite_approximate_1946, pustejovsky_small_2018,
tipton_small-sample_2015}. Under independent study clusters, adequate
cluster count, and the usual CR2 regularity conditions, this reduces
sensitivity of inference to the working $\rho$; it is not an assumption-free
guarantee. Estimation uses fully inverse-variance weights via
\pkg{metafor}'s \code{rma.mv}, as \citet{pustejovsky_meta-analysis_2022}
recommend. Across $\rho \in \{0.3, 0.6, 0.9\}$ each focal coefficient shown
in Appendix~\ref{app:metareg} moves by at most .001; this is an empirical
sensitivity result for those terms, not a general validity guarantee.

\subsection{Robustness analyses}\label{sec:robustness}

Three features of the design double as robustness analyses and are reported
with the results. First, the six-estimator portfolio is a sensitivity
analysis over estimand choice, re-estimated on the common subset of units
carrying standard errors for all six estimators
(Section~\ref{sec:estimator-dependence}). Second, the two-wave completion of
the corpus gives stratum-level replications of the headline
(Section~\ref{sec:composition}). Finally, a residual-covariate audit adds
every remaining harvestable dataset property, including sample frame,
sparsity, provenance grade, model family, and unidimensionality, to test
whether the confirmatory conclusions survive (Appendix~\ref{app:metareg}).

\section{Results}\label{sec:results}

\subsection{The distribution of computed reliability}\label{sec:distribution}

\begin{figure}[t]
\centering
\includegraphics[width=\textwidth]{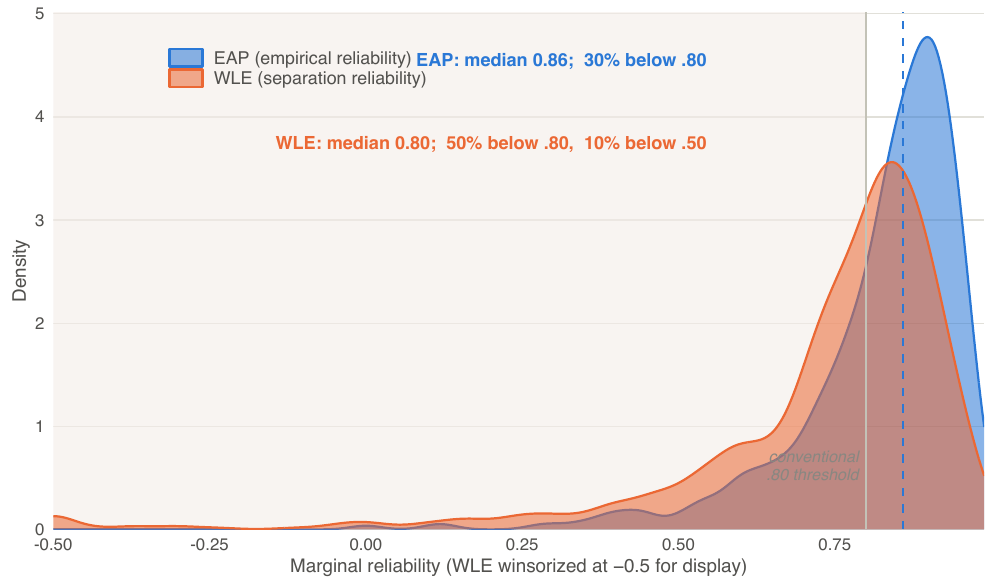}
\caption{The distribution of marginal reliability across 889 administration
units under the two co-primary estimands. Dashed vertical lines mark medians,
the gray line marks the conventional .80 threshold, and the shaded region is
the sub-threshold zone. Under EAP empirical reliability (blue) the median is
.86 and 30\% of units fall below .80. Under WLE separation reliability
(orange) the median is .80, 50\% fall below .80, and 10\% fall below .50, all
observed shares. WLE values below $-0.5$ (9 units, of which 4 are numerically
degenerate below $-1$) are winsorized for display only; all analyses use the
raw values.}
\label{fig:distribution}
\end{figure}

\input{floats/tab2_portfolio}

In \cref{fig:distribution} we show the distribution under the two co-primary
estimands, and \cref{tab:portfolio} summarizes the full portfolio. Under EAP
empirical reliability, the lenient co-primary and the quantity that
\pkg{mirt} users most often report, the median dataset reaches .859, which
appears comfortable. Even here, however, 30\% of datasets fall below .80 and
14\% below .70. Under WLE separation reliability, the strict co-primary and
the quantity a user of unshrunken scores experiences, the median drops to
.801, the sub-.80 share rises to 50\%, and 10\% of datasets fall below .50,
the region in which individual scores carry more noise than signal.
Twenty-one units have negative separation reliability, their realized score
variance being smaller than their average squared standard error, which
happens when an instrument is administered to a population it cannot
distinguish \citep{adams_reliability_2005}.

That last count requires a comment, because an earlier version of this
analysis reported 48 such units. The difference is item-key alignment.
Before items were aligned to a common direction, 27 further units returned
negative separation reliabilities that were artifacts of mixed keying and
not statements about the instruments; Appendix~\ref{app:keyalign} documents
the correction and measures its effect estimator by estimator. The Rosenberg
Self-Esteem Scale is the clearest example. Five of its ten items are
reverse-worded, and on the unaligned matrix its separation reliability was
negative, whereas alignment restored it to .89. Its EAP reliability of .92
was never touched by the defect. Negative separation reliability is a real
and interpretable phenomenon, but it is rarer than an unaligned pipeline
suggests, and a negative value should prompt a keying check before it
prompts a substantive conclusion.

The two co-primaries usually disagree in one direction, as the ordering of
Section~\ref{sec:order} leads one to expect: in 97.1\% of units EAP exceeds
WLE, the 25 reversals have median magnitude .002, and the median within-unit
gap is .045. The implemented gap is a diagnostic rather than a theorem,
since beyond the estimand contrast it also contains calibration and package
differences, together with slope-model differences on polytomous units.

\subsection{Estimation noise and true heterogeneity}\label{sec:real}

\begin{figure}[t]
\centering
\includegraphics[width=\textwidth]{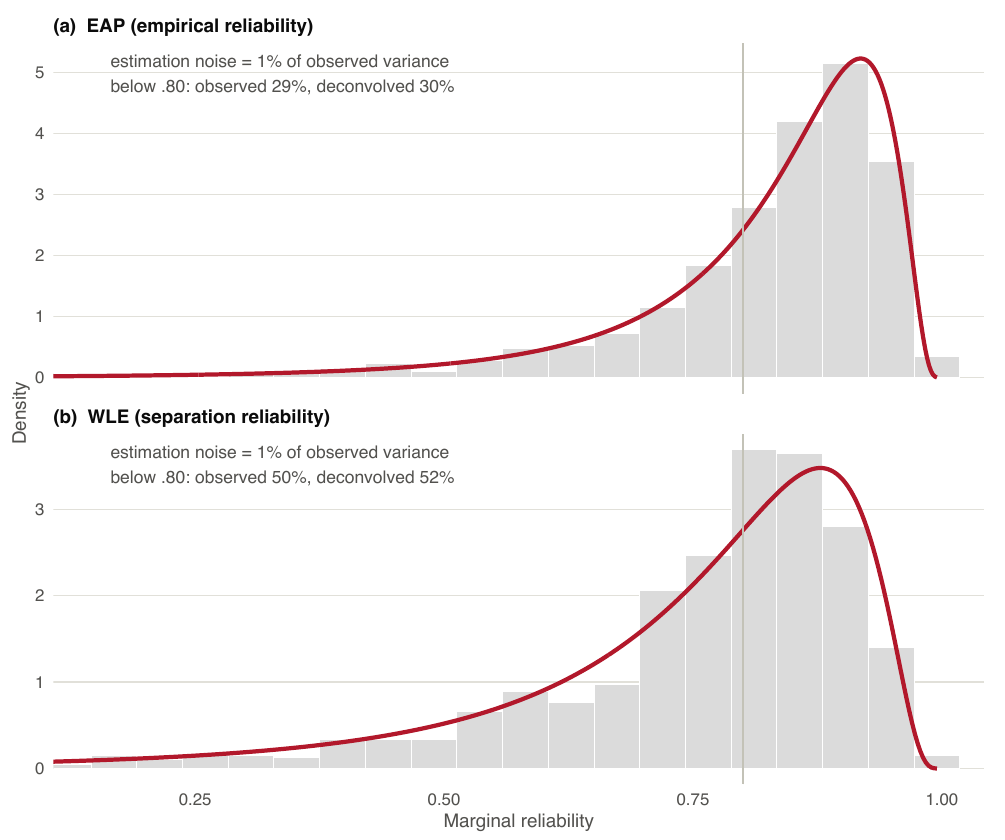}
\caption{Observed and deconvolved distributions of reliability. Gray bars
show observed point estimates, and the red curve is the true distribution
implied by the three-level model of \cref{eq:decon}, which separates
estimation noise (estimated bootstrap variances treated as known) from
between-dataset heterogeneity. Panel (a): EAP empirical reliability
($n = 877$ units with standard errors). Panel (b): WLE separation reliability ($n = 865$, the 867 units with a
standard error less two values below $-1$ excluded as numerically
degenerate). Estimation noise accounts for about 1\% of the observed
variance in both panels. The display range is .15--1; 4 EAP and 23
WLE analysis values below .15 are omitted from the plot only and retained in
every calculation.}
\label{fig:deconv}
\end{figure}

One objection to \cref{fig:distribution} is that the low tail reflects noisy
estimation, with small samples scattering below the threshold by chance. The
deconvolution rejects this. On the $L$ scale, average unit-level sampling
variance is about 1\% of the observed variance, 1.0\% for EAP and 0.8\% for
WLE, so removing it barely moves the distribution. Under deconvolution the
true sub-.80 share is 30.3\% for EAP and 51.8\% for WLE, against observed
shares of 29.1\% and 49.7\%, and the true sub-.50 share under WLE is 11.1\%.
Nor is the qualitative result an artifact of the parametric tail. Gaussian
empirical-Bayes unit predictions give shares of 28.7\% and 49.7\%; these
avoid a closed-form normal tail but retain Gaussian hierarchy and shrinkage.
A flexible spline $g$-model agrees with the normal model to 0.1 points for
EAP and 1.3 points for WLE at the .80 threshold (Appendix~\ref{app:gmodel}).
Heterogeneity is far larger than noise: $\hat\tau_{\mathrm{total}} = .72$ on
the $L$ scale for both estimands, against median unit standard errors an
order of magnitude smaller. Interestingly, more of the heterogeneity lies
within studies than between them ($\tau^2_{\mathrm{unit}} >
\tau^2_{\mathrm{study}}$ for both estimands; Appendix~\ref{app:deconv}).
Reliability is thus a property of the individual scale administration and
not of the laboratory or the paper, which is why citing a validation article
cannot substitute for computing reliability in the sample at hand
\citep{thompson_psychometrics_2000}.

\subsection{Estimator dependence}\label{sec:estimator-dependence}

\begin{figure}[t]
\centering
\includegraphics[width=\textwidth]{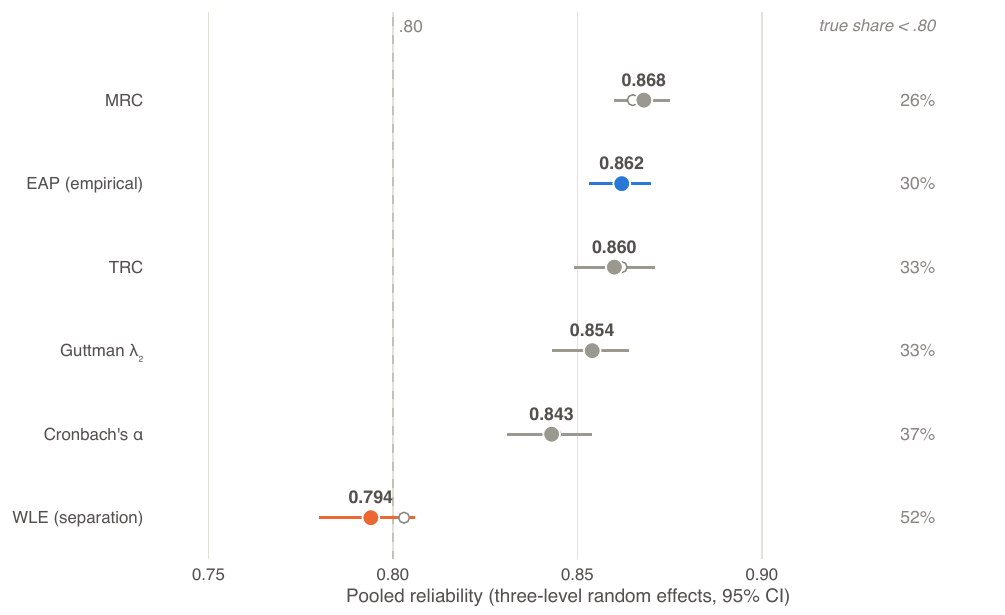}
\caption{Estimator dependence of the corpus-level answer. Pooled three-level
estimates with 95\% confidence intervals are shown for six estimators of the
same corpus, ordered from bottom to top. Filled points use each estimator's
maximal sample and hollow points the common subset carrying standard errors
for all six ($n = 621$--623), where the arrangement is identical. The right
margin gives the deconvolved true share below .80 under each estimator. The
portfolio spans .794--.868 in pooled reliability and 26--52\% in sub-.80
share. The co-primary order (WLE below EAP) is empirically consistent with
the information-functional order but is not its finite-sample corollary, and
$\alpha \le \lambda_2$ follows \citet{guttman_basis_1945}, whereas the
intermediate placement of the anchors and the expected-form auxiliaries is
an empirical regularity of this corpus.}
\label{fig:spec}
\end{figure}

In \cref{fig:spec} we summarize the estimator comparison. Six estimators,
each publishable as the marginal reliability of these data, return pooled
values from .794 for WLE separation to .868 for MRC, and true sub-.80 shares
from 26\% to 52\%. In this corpus the $\wbar$-path estimator sits lowest and
the posterior-path pair highest, an empirical arrangement consistent with
but not entailed by \cref{thm:order}. Within the posterior path the
empirical and expected versions, EAP and MRC, nearly coincide, because
averaging over the realized sample or over $N(0,1)$ matters little when
calibration samples are large. Between them fall the model-free anchors and
the sum-score TRC, with $\alpha \le \lambda_2$ as \citet{guttman_basis_1945}
requires. Note that the theorem does not govern the anchors' placement, so
that part of the arrangement is an empirical regularity and not a necessity.
On the common subset of units carrying all six estimators ($n = 621$--623)
the arrangement is identical, so coverage composition explains none of the
spread.

At the corpus level the choice of estimator therefore accounts for about
seven points of pooled reliability and twenty-six percentage points of the
sub-.80 share, which covers most of the distance between what
reviewers treat as acceptable and what they treat as questionable. In
practice the choice is rarely stated and is usually made by software
default, with \pkg{mirt} users inheriting the posterior path and
Rasch-tradition users the separation path \citep{soland_how_2024,
adams_reliability_2005}. Downstream analyses inherit whatever the default
produces. Parallel work on the warehouse shows that scoring decisions alone,
holding students and items fixed, can move value-added rankings of teachers
and schools by around twenty percentile points
\citep{gilbert_sensitivity_2026}.

Part of an apparent estimator spread can also be manufactured by data
processing. Before the keying correction described in
Section~\ref{sec:corpus}, the portfolio span in \cref{fig:spec} was .098
rather than .075, because mixed keying biased the coefficients computed on
raw sum scores or under fixed slopes, while the free-slope estimators
largely absorbed the reversed items. Roughly a quarter of what would have
been read as estimand disagreement was therefore processing artifact. The
estimator-by-estimator account of that correction is given in
Appendix~\ref{app:keyalign}.

\subsection{Moderators of reliability and of the estimand gap}
\label{sec:moderators}

\input{floats/tab3_metareg}

\begin{figure}[t]
\centering
\includegraphics[width=\textwidth]{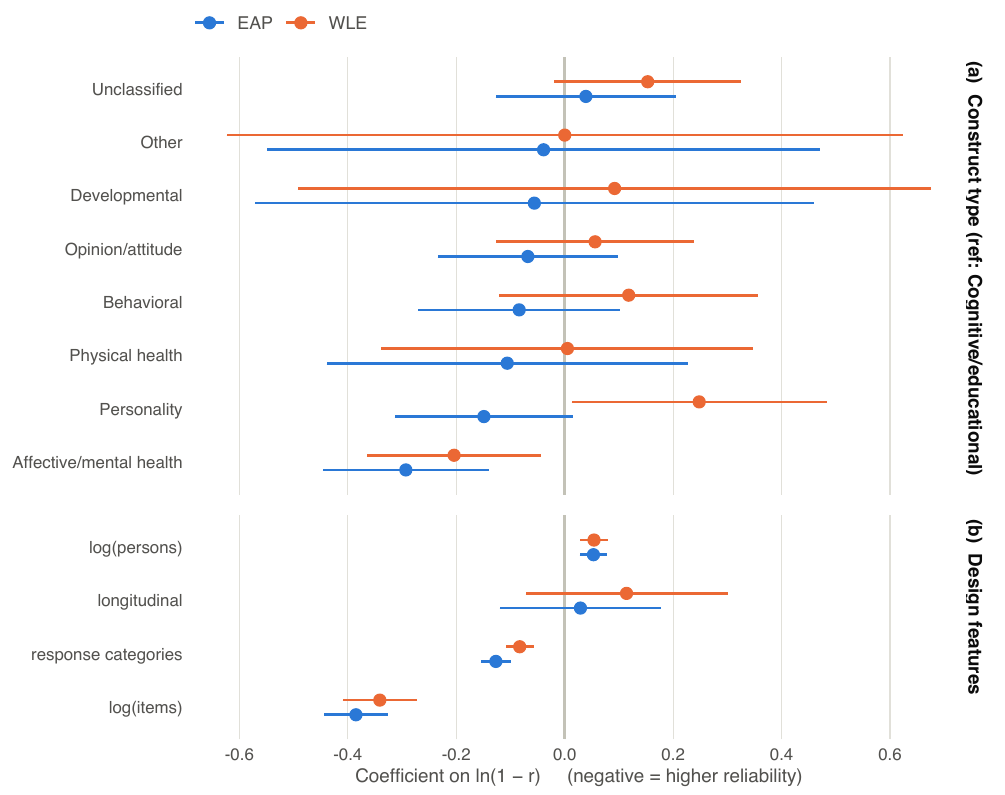}
\caption{Meta-regression coefficients on $\ln(1-r)$ with CR2 cluster-robust
95\% normal-Wald intervals, for the EAP (blue) and WLE (orange) estimands.
Negative coefficients indicate higher reliability. Panel (a): construct-type
contrasts against the cognitive/educational reference; the formal
within-unit test of moderation of the implemented gap appears in
Section~\ref{sec:moderators}. Panel (b): design features. Longer tests and
more response categories raise reliability under both estimands, while
larger samples show slightly lower reliability, which
Section~\ref{sec:composition} identifies as a composition effect.}
\label{fig:metareg}
\end{figure}

\Cref{tab:metareg} and \cref{fig:metareg} summarize the meta-regression.
Three design effects hold under both estimands. Test length dominates, with
$\hat\beta_{\ln J} = -.385$ for EAP and $-.341$ for WLE, both $p < .001$, so
that doubling test length multiplies unreliability $(1-r)$ by about
0.77--0.79. For comparison, the Spearman--Brown prophecy
\citep{spearman_correlation_1910, brown_some_1910} predicts a multiplier of
$1/(1+r) \approx 0.56$ at $r = .80$. Lengthening a test in practice
therefore yields roughly half the gain that the textbook formula predicts.
Since the comparison is between instruments rather than within them, it does
not contradict the Spearman--Brown law, which assumes parallel added items.
Rather, it indicates that real added items bring heterogeneous content, and
it provides a descriptive benchmark for what lengthening accompanies across
instruments; it does not bound a within-test intervention. More response
categories also help, with about 12\% less unreliability per additional
category under EAP, and the repeated-measures indicator is null. Larger
samples show slightly lower reliability under both estimands, but within the
small-sample stratum the slope vanishes, which identifies the pooled
coefficient as a between-stratum composition contrast, a matter of who
collects very large samples, rather than as a within-stratum gradient
(Appendix~\ref{app:metareg}; Section~\ref{sec:composition}).

Construct-type moderation, by contrast, differs across the two co-primary
specifications, and this requires a direct test. Comparing significance
patterns across the two columns of \cref{tab:metareg} would not provide one,
because the difference between significant and non-significant is not itself
significant \citep{gelman_difference_2006}. For each unit carrying both
co-primaries we therefore form the within-unit difference
$\Delta_i = L^{\mathrm{WLE}}_i - L^{\mathrm{EAP}}_i$ and regress it on the
same moderators with CR2 study-clustered errors ($n = 875$ units in 430
study clusters; Appendix~\ref{app:metareg}). Differencing removes additive
study and unit components shared exactly by the two outputs, but not the
package, calibration, or slope-model components, so these coefficients
describe moderation of the implemented WLE--EAP gap rather than of the
estimand contrast alone; we return to this qualification in the limitations.
The construct factor is jointly significant (Hotelling--Zhang test,
\citealp{tipton_small-sample_2015}; $F(8, 25.3) = 2.72$, $p = .026$), and
the personality contrast dominates and survives Holm correction
\citep{holm_simple_1979} over the eight construct contrasts
($\hat\beta_\Delta = +.323$, SE $= .085$, $p_{\mathrm{Holm}} = .004$). In
other words, the WLE-to-EAP unreliability ratio is about 38\% larger for
personality scales than for cognitive tests. Behavioral ($+.186$,
$p_{\mathrm{Holm}} = .030$) and opinion--attitude ($+.133$,
$p_{\mathrm{Holm}} = .021$) contrasts also survive correction, while the
affective contrast does not ($p_{\mathrm{Holm}} = .33$). In
Appendix~\ref{app:metareg} the residual covariate audit suggests a mechanism
for each side. On the EAP scale the affective advantage in
\cref{tab:metareg} ($-.293$) is fully mediated by unidimensionality and
model family, since affective screeners are short, strongly one-factor
Likert instruments and the contrast vanishes once those two variables enter.
The WLE personality penalty is not explained away and in fact strengthens
slightly under the full audit block ($+.248$ to $+.300$, $p = .008$), which
is consistent with personality inventories concentrating information
narrowly, so that the unbounded $\E[1/\Jinfo]$ integral weighs their flat
tails heavily. In summary, construct differences in the level of reliability
are modest, whereas the disagreement between the implementations is
construct-structured. In all, the confirmatory moderators explain 35\% of
total heterogeneity under EAP and 19\% under WLE, so most of what makes a
particular administration reliable is not predictable from design
descriptors and has to be measured.

\subsection{Composition by completion wave}\label{sec:composition}

\begin{figure}[t]
\centering
\includegraphics[width=\textwidth]{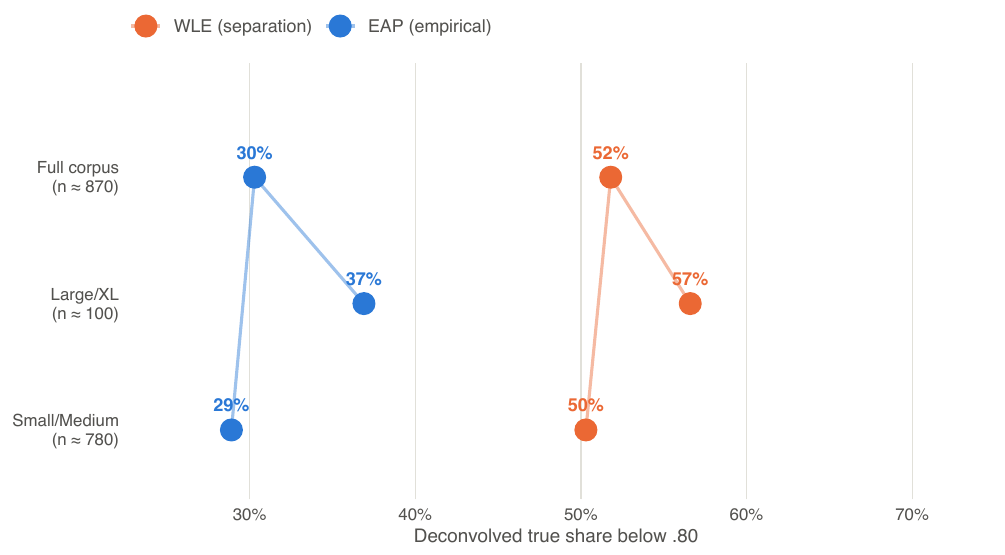}
\caption{Composition of the headline. Deconvolved true share below .80 is
shown for the first completion wave (labeled Small/Medium), the second wave
(labeled Large/XL: the 101 largest datasets, of which 94 and 100 carry
standard errors under EAP and WLE), and the full corpus, under both
co-primary estimands. The second wave, dominated by internet panels and
large surveys carrying short embedded scales, is less reliable than the
first, so completing the corpus raised the headline shares.}
\label{fig:composition}
\end{figure}

A natural objection to early versions of this corpus was compositional. Its
first wave under-represented very large datasets, and professionally built
large-scale tests might be expected to be highly reliable, pulling the
distribution upward. In \cref{fig:composition} the completed corpus answers
the objection empirically. Under both estimands the second wave, the 101
largest datasets described in Section~\ref{sec:corpus}, is less reliable
than the first, with true sub-.80 shares of 36.9\% against 28.9\% for EAP
and 56.6\% against 50.3\% for WLE. The reason is compositional: in a public
warehouse, large datasets are typically internet panels and national surveys
carrying short embedded scales rather than flagship achievement tests.
Including the second wave moved the full-corpus shares from 28.9\% to 30.3\%
for EAP and from 50.3\% to 51.8\% for WLE. The direction of this composition
effect also explains the positive log-persons coefficient in
\cref{tab:metareg}.

\section{Discussion}\label{sec:discussion}

\subsection{Interpreting the two co-primary estimands}\label{sec:reading}

Three empirical statements now have corpus-level support. Sub-threshold
reliability is not an exceptional case: under the strict co-primary estimand
it is the majority outcome, and under the lenient one it covers three
datasets in ten. These differences are real and not estimation scatter, and
they lie mostly within studies, so per-instrument validation citations
cannot substitute for computing reliability in the sample at hand. Finally,
the answer to the question of how reliable a set of scores is depends on the
estimand, by mathematical necessity for the information functionals
(\cref{thm:order}) and by about seven points of pooled reliability in this
corpus, and the dependence extends to moderation, since
the gap between the estimands is itself construct-structured
(Section~\ref{sec:moderators}).

Both co-primaries are best read as complements answering different questions
rather than as competitors. EAP empirical reliability describes the
precision of shrunken scores, which is the relevant quantity when downstream
use also shrinks, as in multilevel modeling or empirical-Bayes ranking. WLE
separation reliability describes unshrunken scores, the relevant quantity
when individual point estimates are used as they stand, as in cut-score
decisions or secondary regressions on observed scores. As a limiting case,
the 21 negative-WLE units are populations their instruments cannot order.
When the two numbers diverge sharply, the divergence itself is informative,
since it indicates that the test's information is concentrated in a narrow
region relative to where its population sits.

\subsection{Implications for reporting practice}\label{sec:implications}

The results motivate three inexpensive changes to routine practice and one
to software defaults.

First, name the estimand. A statement such as ``marginal reliability was
.86'' is under-specified in the same way that an unlabeled effect size would
be. One clause suffices: which branch (separation or posterior; latent or
sum score) and which version (empirical or expected).

Second, attach uncertainty. Reliability coefficients are estimates, and
standard errors are available at modest cost for much of this paper's
portfolio: analytic for $\alpha$ and $\lambda_2$ \citep{ark_standard_2025}
and for MRC and TRC \citep{andersson_large_2018}, and one bootstrap loop for
the rest. A coefficient of .78 with a standard error of .04 licenses
different claims than the same coefficient with a standard error of .005.
This recommendation is not new in spirit. Cronbach's retrospective judgment
was that a coefficient is a crude device and that error information, in his
case the standard error of measurement for scores, is the most important
single thing to report about an instrument \citep{cronbach_my_2004}.

Third, report a strict and a lenient operational coefficient together. A
two-number report, WLE separation alongside EAP empirical reliability,
compares two common scoring paths at the cost of one line and converts the
silent estimator degree of freedom into information. A sharp divergence is a
diagnostic that can reflect targeting, calibration, model, and scoring
differences and should be investigated rather than dismissed.

Fourth, fix the defaults. Most applied users report whatever their software
prints, and the corpus-level spread documented here measures what that
practice currently costs \citep{soland_how_2024}. Package authors could
print the pair, or at least label the branch, at negligible cost.

\subsection{Relation to the literature on reported coefficients}
\label{sec:relation}

Relative to reliability generalization \citep{vachahaase_reliability_1998,
sanchezmeca_recommended_2013}, this audit changes three things. It computes
coefficients instead of harvesting them, which removes reporting selection.
It spans instruments rather than replications of one, and it attaches
unit-level uncertainty, which is what enables the deconvolution. In return
we accept a different selection frame, datasets that reach a public
warehouse, which we discuss under limitations.

We can quantify the computed-versus-reported difference, because
\citet{greco_meta-analysis_2018} report, construct family by construct
family, the share of reported alpha values below each threshold. Their
corpus and ours overlap in three construct families: individual differences,
work attitudes, and behaviors. Across those families their weighted shares
are 6.8\% below .70 and 29.0\% below .80, whereas ours, computed on the same
three construct types in the warehouse ($n = 236$), are 14.8\% and 40.3\%.
Two readings of the gap are available and the data do not fully separate
them. One is composition, since Greco's corpus is built from established
management and organizational instruments while ours contains whatever a
public warehouse holds. The other is reporting selection. Test length, the
strongest moderator in this paper, does not account for the gap, because
restricting our units to at least ten items (median 21, longer than most of
the scales in their corpus) leaves gaps of 7.2 and 5.0 points, and adding a
sample-size floor widens rather than narrows them. We therefore read the
residual gap as a descriptive composition-plus-selection contrast, and not
as a bound or a clean estimate of publication bias. Its direction is
nevertheless consistent and its magnitude large enough to matter, in that a
reader calibrating expectations from the published literature will expect
better measurement than a warehouse of raw data delivers. Note that before
the keying correction this gap measured 26.5 points at the .70 threshold,
and the correction removed 69.8\% of it (Appendix~\ref{app:keyalign}); the
removed portion was our own processing artifact, which is one reason
pipeline auditing should precede comparisons with published coefficients.

A sharper test suggested by \citet{hussey_aberrant_2025} is out of reach.
They detect reporting distortion as excess mass at threshold values, and
computed coefficients, which pass through no reporting filter, should show
none. Our corpus is roughly two orders of magnitude too small for that
comparison (Section~\ref{sec:limitations}, item 10). Relative to the
reliability-paradox literature \citep{hedge_reliability_2018,
rouder_hierarchical_2024}, which explains why well-replicated experimental
effects yield unreliable difference scores, our contribution is breadth.
Their warning generalizes far beyond cognitive-control tasks, and the corpus
supplies the base rates. Relative to the measurement meta-science agenda
\citep{flake_measurement_2020, parsons_psychological_2019}, the corpus
supplies what exhortation lacks, namely a defensible empirical prior on the
reliability to expect when a paper reports none.

\subsection{Limitations}\label{sec:limitations}

Ten limitations bound the claims. (1) \emph{Selection}. Our corpus contains
only datasets whose authors could and did share them. Which data become
publicly analyzable is governed by discretionary disclosure decisions rather
than by sampling \citep[cf.][]{acciai_estimating_2023}, and high-stakes
operational testing programs are under-represented, while the composition
analysis shows that sample size is no proxy for that stratum. What we
describe here is the distribution of shareable psychology, which is also the
psychology most secondary analysis touches. (2) \emph{Model}. Primary
estimates assume unidimensional Rasch or GPCM calibrations. This is the
field's default scoring model, and cross-dataset comparisons on the
warehouse suggest that added dimensions buy little in predictive terms for
large dichotomous datasets \citep{liu_comparing_2026}, but a
multidimensional reanalysis would define different estimands rather than
correct these. (3) \emph{Package and model confounding}. The two
co-primaries come from different packages and, for polytomous units, fixed-
versus free-slope models. Their delta regression therefore concerns the
implemented gap, not the estimand contrast alone. (4) \emph{WLE uncertainty}
is conditional on the calibration (Section~\ref{sec:ses}). Our measured
shortfall is a median standard-error factor of 1.06, and propagating it
changes no headline by more than 0.1 percentage points. (5) \emph{Sampling}
caps precision at the 5{,}000-person level for the 89 sampled large
datasets, which is small next to the heterogeneity. (6) \emph{Construct
labels} are warehouse metadata with 19\% of units unclassified, and we model
the unclassified category explicitly rather than dropping it. (7) \emph{The
item facet is fixed} (Section~\ref{sec:itemfacet}), so every coefficient
conditions on the realized item set. Item-random generalizability may be
lower under a defensible content domain, but the direction and magnitude are
not identified here. (8) \emph{Item-key alignment} is a post hoc corrective
rule frozen before the correction rerun. Known positive controls support it,
but corpus-wide error rates and 11 cognitive/educational acceptances require
an item-level audit (Appendix~\ref{app:keyalign}). (9) \emph{Tail-share
uncertainty}. Deconvolved shares are point estimates conditional on
estimated unit variances and fitted heterogeneity; they do not propagate
uncertainty in the variance components or mixing distribution.
(10) \emph{Threshold distortion} cannot be tested here.
\Citet{hussey_aberrant_2025} document excess reported alphas at .70, .80,
and .90, and the natural comparison, whether computed coefficients show no
such excess, fails for power. With 703--872 values per estimator spread over
60 bins, no permutation test approaches significance (smallest $p = .29$),
and at the .70 threshold the smallest caliper ratio our sample sizes can
distinguish from unity is 2.6--5.5, well above the 1.60--1.71 they observe.
We report the attempt so that the required scale is on record.

\subsection{Concluding remarks}\label{sec:conclusion}

We attempted every eligible dataset in a frozen public warehouse and
computed a portfolio of reliability coefficients with unit-level uncertainty
for the 889 that could be analyzed, then used the resulting corpus to
estimate the distribution of true reliability and its dependence on the
estimand. The distribution is wide, its low tail is genuine, and its
location depends on a choice that most reports leave to software defaults.
None of the remedies suggested here is expensive. A report that names its
estimand, attaches a standard error, and gives a strict coefficient beside a
lenient one costs one additional sentence, and the corpus indicates how much
is left unknown when that sentence is omitted. This same infrastructure now
allows the other assumption on which these coefficients rest, the normality
of the latent distribution, to be examined at the same scale, and that is
the subject of a companion investigation.

%% file: floats/tab1_corpus.tex
\begin{table}[t]
\centering
\caption{The analysis corpus: 889 administration units from the Item Response Warehouse local v1.2 freeze. The warehouse inventory contained 933 eligible records for 930 unique datasets. Excluded were 23 units whose items have more than 15 response categories (slider and visual-analog formats), 17 units with more than 600 items (large, mostly adaptive item banks), and 1 degenerate unit; each exclusion carries one machine-readable reason code. Rows are processing tiers by stored size. The corpus was completed in two waves, the second wave comprising the 101 largest datasets (75 XL-tier and 26 medium-tier); the 89 of these with more than 5{,}000 respondents (14 medium-tier) were analyzed on a fixed 5{,}000-person random sample (Appendix~\ref{app:corpus}). Persons and Items are median (10th--90th percentile). The final column is the share with standard errors for both co-primary estimators.}
\label{tab:corpus}
\footnotesize
\setlength{\tabcolsep}{5pt}
\resizebox{\textwidth}{!}{%
\begin{tabular}{lrllrr}
\toprule
Size stratum & Units & Persons & Items & Categories & Both SEs \\
\midrule
Small (processing tier) & 776 & 711 (198--3,428) & 16 (5--59) & 5 & 97\% \\
Medium (processing tier) & 38 & 9,115 (1,832--32,613) & 64 (22--222) & 5 & 87\% \\
XL ($>$50{,}000 persons; sampled) & 75 & 1,000,000 (98,531--1,000,000) & 52 (8--133) & 2 & 93\% \\
\midrule
All units & 889 & \multicolumn{4}{l}{745 grade A + 144 grade B} \\
\bottomrule
\end{tabular}%
}
\end{table}

%% file: floats/tab2_portfolio.tex
\begin{table}[t]
\centering
\caption{Marginal reliability across six estimators of the same corpus, ordered by median. The observed co-primary order (WLE below EAP) is empirically consistent with the information-functional comparison in \cref{thm:order}, but is not its finite-sample corollary; $\alpha \le \lambda_2$ follows Guttman (1945). The intermediate placement of the anchors and expected-form auxiliaries is an empirical regularity of this corpus (\cref{sec:estimator-dependence}). Bold rows are the co-primary estimands. Coverage ($n$) differs by design: $\alpha$/$\lambda_2$ require near-complete data; MRC/TRC skip units with more than 100 items (Appendix~\ref{app:estimation}). Pooled reliability is the three-level random-effects estimate on $L=\ln(1-r)$ (\cref{sec:deconv}), back-transformed, with 95\% confidence interval.}
\label{tab:portfolio}
\footnotesize
\setlength{\tabcolsep}{3pt}
\resizebox{\textwidth}{!}{%
\begin{tabular}{llrrrcrrrc}
\toprule
 & & & & & & \multicolumn{3}{c}{Share below} & \\
\cmidrule(lr){7-9}
Estimator & Family & Point $n$ & SE $n$ & Median & IQR & .80 & .70 & .50 & Pooled $r$ [95\% CI] \\
\midrule
\textbf{WLE (separation)} & $\bar{w}$ (MSEM) & 879 & 867 & 0.801 & 0.69--0.87 & 50\% & 27\% & 10\% & 0.794 [0.780, 0.806] \\
Cronbach's $\alpha$ & CTT anchor & 718 & 647 & 0.839 & 0.76--0.90 & 36\% & 14\% & 4\% & 0.843 [0.831, 0.854] \\
Guttman $\lambda_2$ & CTT anchor & 718 & 647 & 0.848 & 0.77--0.91 & 31\% & 13\% & 4\% & 0.854 [0.843, 0.864] \\
TRC & expected-form & 825 & 807 & 0.849 & 0.76--0.92 & 34\% & 16\% & 6\% & 0.860 [0.849, 0.871] \\
\textbf{EAP (empirical)} & $\rho_{\mathrm{PSD}}$ & 889 & 877 & 0.859 & 0.78--0.91 & 30\% & 14\% & 4\% & 0.862 [0.853, 0.870] \\
MRC & expected-form & 825 & 805 & 0.863 & 0.79--0.91 & 27\% & 11\% & 1\% & 0.868 [0.860, 0.875] \\
\bottomrule
\end{tabular}%
}
\end{table}

%% file: floats/tab3_metareg.tex
\begin{table}[t]
\centering
\caption{What moves reliability: correlated-and-hierarchical-effects (CHE) meta-regression of $L = \ln(1-r)$ on design features and construct type, with CR2 cluster-robust standard errors (clustering on study). Negative coefficients indicate \emph{higher} reliability. $N = 877$ (EAP) / $865$ (WLE) units in $433/422$ study clusters. Reference construct: Cognitive/educational. $^{*}p<.05$, $^{**}p<.01$, $^{***}p<.001$; df = CR2 Satterthwaite degrees of freedom; $^{\dagger}$ marks contrasts with df $<4$, where small-sample robust inference is unreliable and stars should be discounted. Full diagnostics, $\rho$ sensitivity, and the residual-covariate audit appear in Appendix~\ref{app:metareg}.}
\label{tab:metareg}
\footnotesize
\setlength{\tabcolsep}{7pt}
\begin{tabular}{lcrcr}
\toprule
 & \multicolumn{4}{c}{Coefficient on $\ln(1-r)$, $\hat\beta$ (SE)} \\
\cmidrule(lr){2-5}
 & \multicolumn{2}{c}{EAP} & \multicolumn{2}{c}{WLE} \\
\cmidrule(lr){2-3}\cmidrule(lr){4-5}
Moderator & $\hat\beta$ (SE) & df & $\hat\beta$ (SE) & df \\
\midrule
\multicolumn{5}{l}{\itshape Design features} \\
log(items) & $-0.385$ (0.030)$^{***}$ & 99 & $-0.341$ (0.035)$^{***}$ & 88 \\
log(persons) & $+0.053$ (0.013)$^{***}$ & 150 & $+0.054$ (0.013)$^{***}$ & 157 \\
Response categories & $-0.127$ (0.014)$^{***}$ & 43 & $-0.083$ (0.013)$^{***}$ & 40 \\
Longitudinal & $+0.029$ (0.076) & 44 & $+0.114$ (0.095) & 44 \\
\addlinespace
\multicolumn{5}{l}{\itshape Construct type (ref.\ Cognitive/educational)} \\
\quad Affective/mental health & $-0.293$ (0.078)$^{***}$ & 86 & $-0.204$ (0.082)$^{*}$ & 78 \\
\quad Personality & $-0.149$ (0.084) & 50 & $+0.248$ (0.120)$^{*}$ & 45 \\
\quad Opinion/attitude & $-0.068$ (0.085) & 68 & $+0.056$ (0.093) & 63 \\
\quad Behavioral & $-0.084$ (0.095) & 36 & $+0.118$ (0.122) & 33 \\
\quad Physical health & $-0.106$ (0.170) & 15 & $+0.005$ (0.175) & 15 \\
\quad Developmental & $-0.056$ (0.263) & 5 & $+0.092$ (0.298) & 5 \\
\quad Other & $-0.039$ (0.260) & 9 & $+0.000$ (0.318) & 8 \\
\quad Unclassified & $+0.039$ (0.085) & 121 & $+0.153$ (0.088) & 114 \\
\midrule
$\tau^2$ (study / unit) & 0.101 / 0.235 & & 0.146 / 0.275 & \\
\bottomrule
\end{tabular}
\end{table}

%% file: appendices/appendix_A.tex
\section{Formal results for the estimand family}\label{app:family}

In this appendix we formalize the family of Section~\ref{sec:family}, prove
the total order of \cref{thm:order}, and report its numerical verification.

\subsection{Definitions}\label{app:family-defs}

Let $\theta \sim G$ with $0 < \sigma^2_\theta < \infty$, and let
$\Jinfo(\theta) > 0$ $G$-almost surely be the test information function of a
fixed calibration. Assume $\E_G[\Jinfo] < \infty$ and
$\E_G[1/\Jinfo] < \infty$. The three information functionals are
\begin{align}
\wbar &= \frac{\sigma^2_\theta}{\sigma^2_\theta + \E_G[1/\Jinfo(\theta)]},
\tag{A.1}\label{eq:A-wbar}\\[2pt]
\rho_{\mathrm{PSD}}^{(I)} &= 1 - \frac{1}{\sigma^2_\theta}\,
\E_G\!\left[\frac{1}{\Jinfo(\theta) + 1/\sigma^2_\theta}\right],
\tag{A.2}\label{eq:A-psd}\\[2pt]
\rhotilde &= \frac{\sigma^2_\theta\,\bar{\Jinfo}}{\sigma^2_\theta\,\bar{\Jinfo}
+ 1}, \qquad \bar{\Jinfo} = \E_G[\Jinfo(\theta)].
\tag{A.3}\label{eq:A-tilde}
\end{align}
In \cref{eq:A-psd}, the superscript $(I)$ distinguishes this functional from
the exact posterior reliability $\rhopsd=1-\E[\PSD^2]/\sigma^2_\theta$. It
uses the standard Gaussian/Laplace information approximation
$\PSD^2(\theta) \approx [\Jinfo(\theta) + 1/\sigma^2_\theta]^{-1}$
\parencite{bock_adaptive_1982}; the theorem is stated only for the three
integral forms A.1--A.3. Their empirical counterparts replace the
$G$-integrals with averages over realized scores: WLE separation reliability
$1 - \overline{\mathrm{SE}^2}/\Var(\hat\theta_{\mathrm{WLE}})$
\parencite{warm_weighted_1989, adams_reliability_2005} and EAP empirical
reliability $\Var(\hat\theta_{\mathrm{EAP}})/[\Var(\hat\theta_{\mathrm{EAP}})
+ \overline{\PSD^2}]$ \parencite{chalmers_mirt_2012}.

\subsection{Proof of the total order}\label{app:family-proof}

\noindent\emph{Claim.} $\rhotilde \ge \rho_{\mathrm{PSD}}^{(I)} \ge \wbar$,
with equalities if and only if $\Jinfo$ is $G$-almost surely constant.

\medskip
\noindent\emph{Proof.} Write $c = 1/\sigma^2_\theta > 0$.

\emph{(i) Upper inequality.} The map $J \mapsto (J + c)^{-1}$ is strictly
convex on $(0, \infty)$. By Jensen's inequality,
\[
\E_G\!\left[\frac{1}{\Jinfo + c}\right] \;\ge\; \frac{1}{\bar{\Jinfo} + c},
\]
with equality if and only if $\Jinfo$ is $G$-a.s.\ constant. Since
$\rho_{\mathrm{PSD}}^{(I)}$ decreases in that expectation, and
$(1/\sigma^2_\theta)\cdot(\bar{\Jinfo}+c)^{-1} = c/(\bar{\Jinfo}+c) =
1/(\sigma^2_\theta\bar{\Jinfo}+1)$, we obtain
\[
\rho_{\mathrm{PSD}}^{(I)} \;\le\; 1 - \frac{1}{\sigma^2_\theta}\cdot
\frac{1}{\bar{\Jinfo} + c} \;=\; 1 - \frac{1}{\sigma^2_\theta\bar{\Jinfo} + 1}
\;=\; \frac{\sigma^2_\theta\bar{\Jinfo}}{\sigma^2_\theta\bar{\Jinfo} + 1}
\;=\; \rhotilde.
\]

\emph{(ii) Lower inequality.} Substitute $X = 1/\Jinfo > 0$, so that
$(\Jinfo + c)^{-1} = X/(1 + cX)$ and $\E_G[1/\Jinfo] = \E[X]$. Here the map
$\varphi(x) = x/(1 + cx)$ has $\varphi''(x) = -2c/(1+cx)^3 < 0$ and is
strictly concave. By Jensen's inequality, \[ \E[\varphi(X)] \;\le\;
\varphi(\E[X]) \;=\; \frac{\E[X]}{1 + c\,\E[X]}, \] with equality if and
only if $X$, and hence $\Jinfo$, is constant. Then
\[ \rho_{\mathrm{PSD}}^{(I)} \;=\; 1 -
c\,\E[\varphi(X)] \;\ge\; 1 - \frac{c\,\E[X]}{1 + c\,\E[X]} \;=\; \frac{1}{1
+ c\,\E[X]} \;=\; \frac{\sigma^2_\theta}{\sigma^2_\theta + \E[1/\Jinfo]}
\;=\; \wbar. \]

\emph{(iii) Constant information.} If $\Jinfo \equiv I$, direct substitution
gives $\wbar = \rho_{\mathrm{PSD}}^{(I)} = \rhotilde =
I\sigma^2_\theta/(I\sigma^2_\theta + 1)$.
\hfill$\square$

\medskip
The two applications of Jensen's inequality act on two convexity structures
of the same variance ratio, convex in information and concave in inverse
information, and together they enclose $\rho_{\mathrm{PSD}}^{(I)}$ from both
sides. The proof establishes the sign and equality condition of this Jensen
gap. It does not imply that the gap is monotone in
$\operatorname{Var}\{\Jinfo(\theta)\}$, curvature, or an unspecified index of
information dispersion; such a result would require a separately defined
ordering and proof.

\subsection{Numerical verification}\label{app:family-verify}

Two independent audits accompany the proof. The first is a random-design
sweep of $2 \times 10^5$ random test designs, with item counts from 3 to 60,
two-parameter logistic and generalized-partial-credit items in equal
proportion, lognormal discriminations, and $G$ drawn from normal,
skewed-$\chi^2$, and bimodal-mixture families with $\sigma^2_\theta \in
[0.3, 3]$. All three integrals were evaluated by 161-point quadrature. In
every draw the order held, with no violations at relative tolerance
$10^{-9}$, and the minimum observed gaps, of order $10^{-8}$, occurred for
near-flat information designs, as the equality condition requires. This
audit verifies the integral-form algebra. Whether the ordering transfers to
the corpus estimators in finite samples is an empirical matter, and
Section~\ref{sec:distribution} shows that it holds in 97.1\% of units. The
second audit consists of six worked scenarios that trace the estimands under
controlled distortions of information concentration, prior misspecification,
and restricted sampling. They reproduce the qualitative signatures observed
in the corpus, including negative $\wbar$ under severe mismatch between the
information function and the population, where $\wbar = -0.42$ in the
constructed scenario and the corpus contains an analogous empirical case
with WLE separation reliability $-0.31$. Note that the reliability of
estimates used here, which belongs to the parallel-forms correlation
lineage, is not the same as the squared true-score correlation, and the two
definitions need not agree for biased estimators \parencite{kim_note_2012}.
In this corpus the difference is negligible.

\subsection{Expected versus empirical versions}\label{app:family-versions}

For each latent-scale estimand, the expected version integrates over an
assumed $g$, invariably $N(0,1)$, and the empirical version averages over
the realized sample \parencite{zhang_empirical_2006}. Expected versions are
reproducible from parameters alone but inherit the normality assumption,
whereas empirical versions are assumption-light and reflect the realized
sample, including range restriction and adaptive selection. Both faces
appear in the corpus results. EAP in its empirical version and MRC in its
expected version nearly coincide at the corpus level (\cref{tab:portfolio}),
because calibration samples are large and posterior shrinkage limits tail
behavior. For $\wbar$-path quantities, by contrast, the same choice matters
greatly, since the unbounded integrand amplifies any mismatch between
assumed and realized ability distributions \parencite{lee_reliability_2025}.

%% file: appendices/appendix_B.tex
\section{Corpus construction and processing}\label{app:corpus}

\subsection{Snapshot and unit of analysis}

All analyses use a frozen local snapshot of the Item Response Warehouse
\parencite{domingue_introduction_2025}, version 1.2, acquired April 20--22,
2026 and treated as read-only thereafter. Each warehouse table is
materialized in the common long format, with columns \code{id}, \code{item},
and \code{resp}, plus wave identifiers where present. The snapshot combines
two collections: 745 tables from the canonical warehouse release and 144
provisionally ingested tables. A provenance indicator distinguishing the two
enters the residual-covariate audit and is null (\cref{app:metareg}), so the
collections are pooled throughout.

The unit of analysis is the administration unit, one response matrix from
one table--wave combination. For tables with repeated administrations,
detected through wave, time, or session columns, a documented reduction rule
selects one administration. If a wave column fully explains duplicate
person--item cells, the most complete wave is retained, which applies to 93
units. If duplicates remain unexplained, the first observation per cell is
used and the unit is flagged, which applies to 80 units. Remaining units,
716 in all, are single administrations. Administration policy is carried as
a moderator-eligible flag, and the repeated-measures indicator in the main
meta-regression marks units originating from any repeated-measures table.

\subsection{Eligibility and caps}

Eligibility was fixed before the corpus run and required at least 3 items,
at least 2 observed response categories, and at least 100 respondents. The
inventory contains 933 eligible records; three datasets appear in both the
canonical and the provisional collection, and keeping one copy of each
leaves 930 unique datasets. Two caps exclude formats outside the estimand's
scope, namely units with more than 15 observed categories per item (slider
and visual-analog formats) and units with more than 600 items (large,
mostly adaptive item banks). Because some units meet both descriptions, the
released machine reason code uses the executed priority rule: 23 units
receive the category-cap code and 17 the item-bank code. One further unit
was degenerate, having two items after cleaning. Of the 930 unique eligible
units attempted, 889 (95.6\%) produced valid primary estimates. Every
exclusion carries a machine-readable reason code in the released ledger,
and there were no unexplained failures.

\subsection{Response processing}

Item responses are densely recoded per item to $0, 1, \ldots, K_j - 1$, as
the GPCM likelihood requires. Items with a single observed category are
dropped, as are persons with no responses. Missing responses are treated as
ignorable at the scoring stage, through marginal maximum likelihood with a
missing-at-random likelihood contribution. The model-free anchors $\alpha$
and $\lambda_2$ require complete cases and are reported only where item-level
missingness is below 5\% and at least 100 complete cases remain, a
deliberately conservative gate given the sensitivity of sample coefficients
to missingness \parencite{enders_impact_2004}.

\subsection{Item-key alignment}\label{app:keyalign}

The IRW data standard guarantees that each item's response can be treated as
ordinal \parencite{domingue_introduction_2025}. It does not guarantee that
items within a table share a direction, and it should not, since a
reverse-worded item is ordinal on its own scale and flipping it is an
analytic decision rather than a standardization step. For a consumer of the
warehouse the consequence is that any coefficient defined across items must
align item keys itself. Our pipeline initially omitted this step; this
appendix documents the corrective rule, its validation, and the effect of
the correction estimator by estimator.

\subsubsection{The rule and its validation}

We align item keys once, immediately after dense recoding and before any
estimator runs, so that every coefficient sees the same matrix. The rule is
a post hoc corrective rule, introduced after the mixed-keying defect was
found and frozen before the correction rerun; it is algorithmic and does not
use an IRT outcome to decide acceptance. For each unit we form the
pairwise-complete item correlation matrix, treating cells whose item pair is
co-observed fewer than 30 times as uninformative. We then take the leading
eigenvector, orient it so that the majority of clearly loading items
($|v_j| \ge .15$) is positive, and reverse-score the items with
$v_j \le -.15$. Finally, the step is self-validating, in that the alignment
is accepted only if it raises the mean inter-item correlation by at least
.03 and is otherwise discarded entirely.

The acceptance test is not decorative, since two simpler rules failed on
this corpus and were abandoned. The first flipped items whose mean
correlation with the remaining items was negative. When reverse-worded items
make up half the scale, each item's average correlation is close to zero, so
the rule detects nothing, and it recovered none of the Rosenberg scale's
five reverse-worded items. On some units it instead flipped every item, an
operation that leaves $\alpha$ unchanged. The second rule took the leading
eigenvector's sign without an acceptance test, and it produced a false
positive on a national achievement test administered in a matrix-sampled
booklet design. There, with 37.5\% missingness, the pairwise correlation
matrix is nearly uninformative, the leading eigenvalue accounting for 21.2\%
of the total and the median loading being zero, and a handful of large
positive loadings determined the orientation for 55 of 72 items. The
acceptance test rejects this case, because reversing genuinely uninformative
items cannot raise the mean inter-item correlation.

We validated the final rule against instruments whose reverse-item counts
are published. It recovers 5 of 10 for the Rosenberg Self-Esteem Scale, 2 of
4 for the four-item Perceived Stress Scale, 3 of 6 for the Brief Resilience
Scale, 3 of 6 for the Life Orientation Test, and 5 of 10 for the
extraversion marker scale, in each case matching the published key. These
known positives support face validity but do not estimate corpus-wide
sensitivity or specificity. Eleven accepted units carry the broad
cognitive/educational metadata label and require item-level adjudication;
blanket claims that this class is untouched are therefore unwarranted.
Alignment changes the response matrix for 141 of 889 units (15.9\%), and the
remaining units are bit-identical to the unaligned pipeline, so their
estimates are unchanged by construction.

\subsubsection{Effect of the correction, estimator by estimator}

\input{floats/tabS_keyalign}

Not all six coefficients had degraded alike, and \cref{tab:keyalign}
summarizes the correction. An estimator that fits free item slopes can often
absorb a reversed item by assigning it a negative discrimination, which
makes the latent scale comparatively stable, whereas an estimator defined on
the raw sum score, or fitted under a fixed slope, has no such freedom. Panel
A shows the mechanism on a single instrument: on a ten-item extraversion
scale with five reverse-worded items, the free-slope calibration returned
exactly five negative slopes, and after alignment the same slopes reappeared
with identical magnitudes and positive signs, with EAP reliability unchanged
at .896 in either case. Panel B shows the corpus pattern: across the 141
affected units, the median change from alignment was .000 for EAP and .000
for MRC, against $+.205$ for $\lambda_2$, $+.284$ for WLE separation,
$+.378$ for $\alpha$, and $+.396$ for TRC. On the non-aligned control units
every estimator moved by exactly zero. MRC and TRC provide the sharpest
median contrast, since they come from the same \pkg{mirt} calibration and
differ only in the scale on which reliability is defined: the latent-scale
member was median-stable while the sum-score member recovered by .396.

The single-instrument example is not universal, and median stability should
not be read as immunity. Across the corpus, EAP changes by more than .001 in
18 of the 141 units, with four gains above .10 and a range from $-.030$ to
$+.833$. Which branch a coefficient sits on helps determine its failure
modes as well as its value, and a defect can look small in an aggregate
precisely because a free-slope estimator often absorbs sign reversal; the
EAP median showed no symptom while four companion estimators were much more
strongly affected. This is one argument for reporting a portfolio rather
than a single coefficient.

The correction also propagated into the corpus-level results. The
specification-curve span of Figure~\ref{fig:spec} in the main text was .098
before
alignment and is .075 after, so roughly a quarter of what would have
appeared as estimand disagreement was processing artifact. Negative-WLE
counts fell from 48 to 21, and most of an apparent gap between our computed
alpha distribution and the reported alphas of
\textcite{greco_meta-analysis_2018} disappeared (Section~\ref{sec:relation}
of the main text). Three of the realigned second-wave units were originally compared
across different 5{,}000-person samples; the released before/after
comparison refits both states on one deterministic person frame, with
identical person hashes, and archives the residual difference from the
historical export alongside the result. For warehouse analytics in general
the lesson is procedural: coefficients defined across items require key
alignment, and the alignment rule needs validation against instruments with
published keys.

\subsection{Large-sample policy}

The second completion wave contains the 101 largest datasets, namely all 75
units in the XL processing tier and the 26 largest units in the medium
tier. Every second-wave unit with more than 5{,}000 respondents is analyzed
on a fixed random sample of 5{,}000 persons: all 75 XL-tier units and 14
medium-tier units, 89 units in total. The sample is drawn once per unit with
a deterministic seed derived from the unit identifier. Reliability
coefficients are population functionals, so subsampling persons changes
only estimation precision. The induced sampling error, roughly .003--.006
on the reliability scale and captured by the person bootstrap, which
resamples within the same 5{,}000-person frame, is an order of magnitude
below the between-unit heterogeneity of interest ($\hat\tau \approx .7$ on
the $L$ scale). Because the same sampled matrix feeds both co-primary
estimators and both auxiliary and anchor families, estimator contrasts
within a unit are never confounded by the sampling frame.

\subsection{Two-wave completion}

The corpus was computed in two waves under the same eligibility rules and
primary estimator definitions: first the 788 smaller datasets, then the
second wave of the 101 largest, which includes all internet-panel and
administrative datasets. Section~\ref{sec:composition} of the main text
reports the composition analysis this design permits, and the wave-level deconvolution appears in
\cref{tab:strata}; of the second wave, 94 units carry an EAP standard error
and 100 a WLE standard error, which is why the composition analysis counts
differ from 101. All 101 second-wave units converged, with weighted
likelihood scoring enabled throughout, and the largest calibrated unit has
379 items. The post hoc item-key rule was subsequently applied in one
correction round to both waves.

%% file: floats/tabS_keyalign.tex
\begin{table}[H]
\centering
\caption{Item-key alignment. Panel A: fitted GPCM slopes for a ten-item extraversion marker scale with five reverse-worded items, before and after alignment. The model absorbs the reversal as negative discriminations, and alignment restores their sign at similar magnitude. Panel B: median change across the 141 units whose response matrix alignment altered, by estimator; rows carry each estimator's own coverage among those units (141 for the co-primaries, 137 for MRC and TRC, 130 for the complete-case anchors). Free-slope estimators are median-stable with consequential exceptions, whereas raw-sum-score coefficients and the fixed-slope WLE implementation generally increase.}
\label{tab:keyalign}
\footnotesize
\begin{tabular}{lrrr}
\toprule
\multicolumn{4}{l}{\emph{Panel A. Fitted slopes, \texttt{bfi\_goldberg\_1992\_extraversion}}} \\
\midrule
 & Negative slopes & EAP $\rho$ & \\
Before alignment & 5 of 10 & 0.896 & \\
After alignment (5 items reversed) & 0 of 10 & 0.896 & \\
\addlinespace
\multicolumn{4}{l}{\emph{Panel B. Median change on the 141 realigned units}} \\
\midrule
Estimator & Item slopes & Median before & Median after \\
\midrule
EAP (empirical) & free (mirt) & 0.866 & 0.866 \\
MRC & free (mirt) & 0.865 & 0.865 \\
$\lambda_2$ & sum score & 0.590 & 0.835 \\
TRC & sum score & 0.396 & 0.837 \\
WLE (separation) & fixed (TAM PCM) & 0.444 & 0.796 \\
$\alpha$ & sum score & 0.398 & 0.820 \\
\bottomrule
\end{tabular}
\end{table}

%% file: appendices/appendix_C.tex
\section{Estimation details}\label{app:estimation}

\subsection{Model dispatch and software}

Each unit receives one primary calibration by marginal maximum likelihood,
the Rasch model for the 231 dichotomous units
\parencite{rasch_probabilistic_1960} and the generalized partial credit model
for the 658 polytomous units \parencite{muraki_generalized_1992,
masters_partial_1982}. Software versions are \pkg{mirt} 1.46.1
\parencite{chalmers_mirt_2012} with an EM iteration ceiling of 2{,}000
cycles, \pkg{TAM} 4.3-25 \parencite{robitzsch_tam_2026} with a ceiling of 500
iterations for the weighted likelihood scoring, and R 4.6.0
\parencite{rcoreteam_2026}. Every unit carries a deterministic seed, a 32-bit
hash of its identifier plus a global constant, so that all stochastic steps,
including XL person sampling and bootstrap resampling, are exactly
reproducible.

\subsection{The six coefficients}

From the \pkg{mirt} calibration we compute EAP empirical reliability,
$\Var(\hat\theta)/[\Var(\hat\theta)+\overline{\PSD^2}]$ over realized EAP
scores, and, for units with at most 100 items, the expected-form MRC and TRC
via \pkg{irtreliability} \parencite{andersson_irtreliability_2022,
andersson_large_2018}, which require a successful item-parameter covariance
matrix from an observed-information refit. From the \pkg{TAM} calibration we
compute WLE separation reliability,
$1 - \overline{\mathrm{SE}^2}/\Var(\hat\theta)$ over realized
weighted-likelihood scores \parencite{warm_weighted_1989}. From raw complete
cases we compute $\alpha$ ($=\lambda_3$) and $\lambda_2$
\parencite{guttman_basis_1945, cronbach_coefficient_1951}. Coverage by
coefficient is reported in \cref{tab:portfolio}, and every non-covered unit
carries an explicit reason flag in the released dataset, either a missingness
gate, an item-count gate, or a covariance-matrix failure.

\subsection{Cross-package agreement}

Because the co-primary contrast is computed across two packages, we examined
where the packages agree on the quantity they share. An early pilot found EAP
reliabilities from \pkg{mirt} and \pkg{TAM} on identical matrices to agree to
a median absolute difference below .001, and we initially reported that
figure as corpus-wide. Recomputed across all 879 units carrying both values,
the median absolute difference is .0139, and it decomposes cleanly by model:

\begin{center}
\footnotesize
\begin{tabular}{lrrr}
\toprule
Calibration & Units & Median $|\Delta|$ & $|\Delta| > .01$ \\
\midrule
Rasch (dichotomous; same model in both packages) & 224 & $.0004$ & 2\% \\
Polytomous (\pkg{mirt} GPCM vs.\ \pkg{TAM} PCM) & 655 & $.0293$ & 73\% \\
\bottomrule
\end{tabular}
\end{center}

The pilot was dominated by dichotomous units, where the packages fit the same
Rasch model and agreement is essentially exact. On polytomous units they do
not fit the same model, since \pkg{TAM}'s default constrains the slope equal
across items whereas \pkg{mirt}'s GPCM frees it, with fitted slopes on one
ten-item scale ranging from 0.90 to 2.41. Close agreement was therefore never
to be expected there.

We regard this as a property of the design rather than as a defect to be
removed. Separation reliability belongs to the Rasch tradition and is
conventionally computed in fixed-slope models, while posterior-based
reliability is conventionally computed under free-slope calibrations, so
reporting the pair compares the two estimands as they are actually produced.
To assess the part of the contrast associated with the model rather than the
estimand, we refitted \pkg{TAM} with free slopes (\code{irtmodel =
"GPCM"}) on 30 polytomous units stratified by item count and sample size.
Across those units the median change in WLE separation reliability is
$+.009$, from .822 to .842, but the dispersion is not negligible: 23\% of
units move by more than .05, 2 units cross the .80 threshold, and one
degenerate unit moves from $-2.67$ to $-9.93$. This small, cost-selected
diagnostic suggests limited median movement but does not identify a
corpus-wide model-only effect; individual units can be sensitive.

Separately, the full pipeline reproduces an independent hand-built 2025
analysis of 19 datasets with correlation 1.000 on WLE separation reliability,
including an empirical negative case of $-0.31$. That audit is archived with
the pilot outputs in the replication package. Note also that the quantity
\pkg{mirt} labels WLE reliability is a different functional, obtained by
plugging WLE scores into the posterior-path formula, which is why WLE
separation reliability is taken from \pkg{TAM}. This is the kind of naming
collision that Section~\ref{sec:family} is intended to prevent.

\subsection{Compute}

The corpus computes on commodity resources. Point estimation for the small
and medium strata took about 40 core-minutes on a laptop-class machine, and
the bootstrap standard errors and the second (large-dataset) wave together
took 4.6 hours on a 32-vCPU cloud instance, with a resumable ledger that isolates
failures at the unit level. Total cloud cost was under \$15, which indicates
that attaching standard errors to reliability coefficients routinely is not a
question of resources.

%% file: appendices/appendix_D.tex
\section{Standard errors and their cross-validation}\label{app:ses}

\subsection{Bootstrap design for the co-primaries}

For EAP reliability we use a full-refit person bootstrap, in which each
replicate resamples persons with replacement and re-estimates the entire
model, calibration and scoring alike, so that the standard error reflects
both calibration and sampling uncertainty. Replicate counts are cost-aware.
A target $B$ of 100, 60, or 40 by size stratum is reduced when a unit's
measured single-fit time implies that the target would exceed a 240-second
per-unit budget, with a hard floor of 12 \emph{attempted} replicates. Failed
refits are discarded, and the successful count is recorded per unit in the
released dataset. Refits use a reduced EM ceiling
of 500 cycles, which pilot runs showed to alter replicate reliabilities by
less than .001 relative to the full ceiling. The successful-$B$ distribution
has median 100 and lower quartile 40; 19.9\% of units sit below $B = 30$, 36
fall below 12, and 11 have at most three successful refits and therefore no
EAP standard error. Among units with a usable estimate but fewer than 30
successful refits, a single $v_{si}$ carries Monte Carlo error of roughly
$1/\sqrt{2(B-1)}$, or 13--21\% of the standard error. This noise is
unit-local and mean-zero on the variance scale, and aggregated quantities
depend on $\bar v$, whose share of total dispersion is about 1\%, so it is
immaterial downstream.

For WLE separation reliability, re-estimating \pkg{TAM} on every replicate is
computationally prohibitive at corpus scale, at tens of seconds to minutes
per fit on polytomous units. Each unit is therefore calibrated once and the
realized pairs $(\hat\theta_i, \mathrm{SE}_i)$ are resampled with $B = 200$,
recomputing $1 - \overline{\mathrm{SE}^2}/\Var(\hat\theta)$ per replicate.
This standard error is conditional on the calibration, capturing
person-sampling variation but not item-parameter uncertainty, and is therefore a
lower bound. Two considerations support the design. At the corpus
median of $N = 909$, calibration uncertainty is a second-order contribution
to score-based reliability. Moreover, the headline deconvolution conclusion
is insensitive to it, since doubling the WLE sampling variances would raise
the noise share of observed variance only from 0.8\% to about 1.6\%, leaving
more than 98\% of the dispersion attributable to genuine heterogeneity.
Median bootstrap standard errors are .006 for EAP and .010 for WLE on the
reliability scale.

\subsection{Calibrating the WLE lower bound}

Because the pair bootstrap conditions on the calibration, we measured the
omitted share directly. We attempted 30 cross-validation units and obtained
29 usable paired comparisons from
full-refit \pkg{TAM} person bootstraps with $B = 60$, every replicate
re-estimating item parameters and rescoring. Their full-to-pair
standard-error ratios have median 1.06, interquartile range 0.95--1.20, and
range 0.092--1.688, implying a median variance factor of 1.12. Propagating
this factor through the deconvolution changes the model-based noise share
from 0.79\% to 0.89\% and no true-share estimate by more than 0.1 percentage
points. Even an arbitrary factor of 2 leaves all conclusions intact
(\cref{app:deconv}).

\subsection{Analytic standard errors}

The model-free anchors use the multinomial delta-method standard errors of
\textcite{ark_standard_2025}, computed from the published implementation,
vendored verbatim with attribution, and gated to complete-case $n \le
5{,}000$ and $J \le 100$ for cost. MRC and TRC use the item-parameter delta
method of \textcite{andersson_large_2018} via \pkg{irtreliability}, which
requires the calibration's parameter covariance matrix.

\subsection{Cross-validation}

\input{floats/tabS_crossval}

\begin{figure}[H]
\centering
\includegraphics[width=0.92\textwidth]{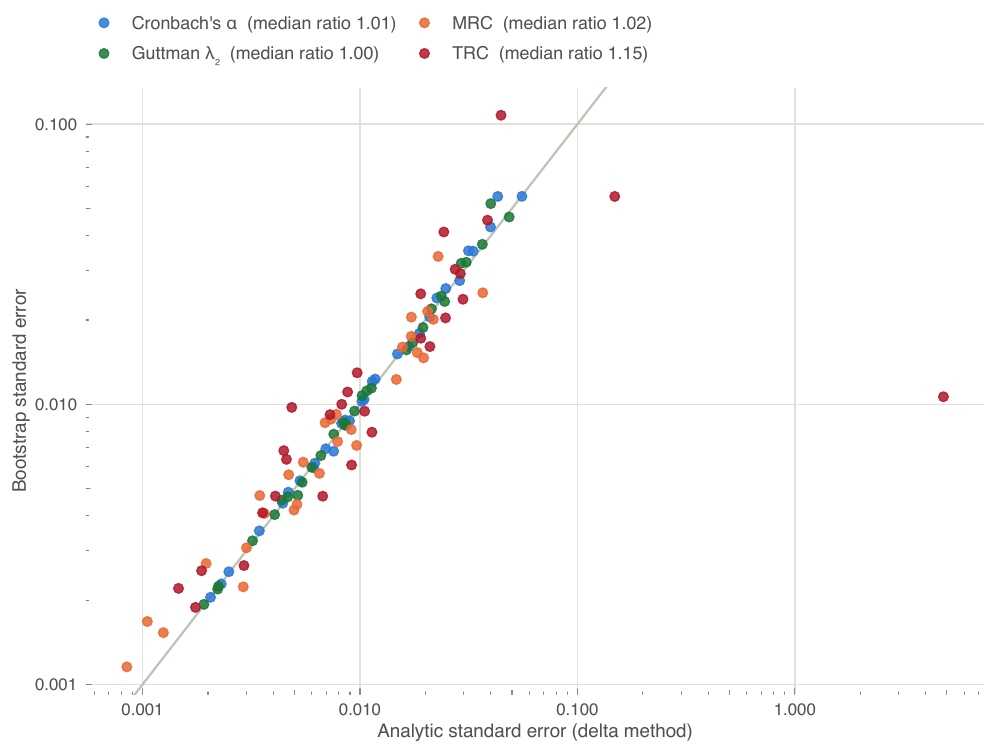}
\caption{Bootstrap versus analytic standard errors on the 30-attempt
cross-validation sample (29 usable for several summaries), on a log--log
scale with the diagonal marking
equality. The classical-test-theory coefficients sit on the line, whereas MRC
and especially TRC scatter mildly above it.}
\label{fig:crossval}
\end{figure}

Thirty units were selected by stratified sampling on item format and size,
among units carrying both analytic-standard-error families and admitting
fast refits ($J \le 60$); the usable rows comprise 26 small-tier units, one medium-tier, and two
XL-tier
units, so large and slow units are severely underrepresented. For $\alpha$
and $\lambda_2$ we ran a $B =
1{,}000$ person bootstrap on the same complete-case frame as the delta
method, and for MRC and TRC a cost-aware bootstrap with $B \in [20, 50]$ in
which every replicate re-runs the full covariance-matrix pipeline.
\Cref{tab:crossval} and \cref{fig:crossval} summarize the outcome. The van
der Ark delta method is essentially exact, with median ratios of 1.006 and
1.003 and log-standard-error correlations of .998. The item-parameter delta
method under-covers mildly, with median ratios of 1.02 for MRC and 1.15 for
TRC. Note that this gap does not diminish with sample size, since the Spearman
correlation between ratio and $N$ is in fact positive, at $+.23$ for MRC.
That identifies its source as uncertainty the delta method conditions away,
namely model-specification propagation and the fixed $N(0,1)$ prior, rather
than as small-sample bias in the asymptotics. All four coefficients fall
inside the pre-specified acceptance band $[0.80, 1.25]$ on the median, but
TRC only on the median: 38\% of its units lie in band, and its
log-standard-error correlation with the bootstrap is .64, against .998 for
the two classical-test-theory anchors and .98 for MRC. Accordingly, we use
bootstrap standard errors for the co-primaries, treat the delta-method
auxiliaries as mildly anticonservative, and regard TRC's standard error as
the least trustworthy quantity in the portfolio. This ranking has no
practical consequence downstream, because auxiliary sampling variances are
dominated by heterogeneity in every pooled analysis ($v_{si} \ll \tau^2$).

%% file: floats/tabS_crossval.tex
\begin{table}[H]
\centering
\caption{Cross-validation of analytic standard errors against person bootstrap: 30 units were attempted and 29 supplied usable MRC/TRC bootstrap pairs. Ratio $=$ SE$_{\text{boot}}$/SE$_{\text{analytic}}$; the reference band for the median ratio was $[0.80, 1.25]$.}
\label{tab:crossval}
\footnotesize
\begin{tabular}{lrrrrr}
\toprule
Coefficient & $n$ & Median ratio & IQR & \%\ in band & $r$(log SE) \\
\midrule
Cronbach's $\alpha$ (vdA delta) & 30 & 1.006 & 0.993--1.043 & 97\% & 0.998 \\
Guttman $\lambda_2$ (vdA delta) & 30 & 1.003 & 0.977--1.029 & 97\% & 0.998 \\
MRC (A--X delta) & 29 & 1.023 & 0.856--1.203 & 69\% & 0.976 \\
TRC (A--X delta) & 29 & 1.150 & 0.822--1.335 & 38\% & 0.639 \\
\bottomrule
\end{tabular}
\end{table}

%% file: appendices/appendix_E.tex
\section{Deconvolution details}\label{app:deconv}

\subsection{Model and identification}

The three-level model of \cref{eq:decon} decomposes each observed
$\hat L_{si} = \ln(1 - \hat r_{si})$ into a corpus mean $\mu$, a study
deviation $u_s \sim N(0, \tau^2_{\mathrm{study}})$, a within-study unit
deviation $u_{si} \sim N(0, \tau^2_{\mathrm{unit}})$, and sampling noise
$\varepsilon_{si} \sim N(0, v_{si})$ with the estimated $v_{si}$ treated as
known, the squared bootstrap standard error on the $L$ scale. Because
$v_{si}$ is supplied rather than estimated jointly, restricted maximum
likelihood can separate the
noise layer from the two heterogeneity layers. This is the classical
known-variance empirical-Bayes structure
\parencite{konstantopoulos_fixed_2011, vandennoortgate_threelevel_2013}, that
is, a parametric deconvolution in the sense of
\textcite{efron_deconvolution_2016}. Estimation uses \pkg{metafor}
\parencite{viechtbauer_conducting_2010} with sparse matrices.

\subsection{Back-transformation}

Writing $\hat\tau^2 = \hat\tau^2_{\mathrm{study}} + \hat\tau^2_{\mathrm{unit}}$,
the implied true distribution of $L$ is $N(\hat\mu, \hat\tau^2)$, so the true
share below a reliability threshold $c$ is
\begin{equation}
\widehat{\Pr}(r_{\mathrm{true}} < c) \;=\;
1 - \Phi\!\left(\frac{\ln(1-c) - \hat\mu}{\hat\tau}\right),
\end{equation}
using the monotone-decreasing map $r < c \iff L > \ln(1-c)$, and the true
density on the reliability scale is
\begin{equation}
f(r) \;=\; \frac{1}{\hat\tau\,(1-r)}\;
\phi\!\left(\frac{\ln(1-r) - \hat\mu}{\hat\tau}\right),
\end{equation}
which gives the red curves of \cref{fig:deconv}.

\subsection{Why the transformed scale}

Three considerations select the scale. First, the reliability scale is
bounded above at 1, so no normal mixing model is tenable there, whereas $L$
removes the ceiling. Second, WLE separation reliability can be negative,
which $L$ accommodates ($r < 0 \iff L > 0$) while logit-type transforms
cannot. Third, $L$ is the scale on which the meta-analysis of reliability
coefficients has established variance-stabilized machinery
\parencite{bonett_sample_2002, bonett_varying_2010}, with the delta rule
$\mathrm{SE}_L = \mathrm{SE}_r / (1 - r)$. Comparative work in the
reliability generalization literature supports the choice
\parencite{lopez-ibanez_reliability_2024}.

\subsection{Results and noise decomposition}

\input{floats/tabS_deconv}

\Cref{tab:deconv} reports the full fit. Two noise-share summaries agree. One
is the moment ratio $\bar v/\Var(\hat L)$, which gives $\bar v = 0.0050$
against observed variance $0.494$ for EAP, or 1.0\%, and $0.0040$ against
$0.519$ for WLE, or 0.8\%. The other is the model-based share $\bar
v/(\hat\tau^2 + \bar v)$, which gives 0.95\% and 0.79\%. Note that the REML
components need not reproduce the naive moment identity $\Var(\hat L) =
\tau^2 + \bar v$ exactly, since precision weighting and the hierarchical
structure break it, and for EAP $\hat\tau^2$ (0.519) slightly exceeds the
unweighted variance. On either accounting, noise is about 1\% of dispersion,
so the observed and true distributions are qualitatively close. Exact
threshold movements are not uniformly within one percentage point: EAP below
.80 changes by 1.22 points, WLE below .80 by 2.06 points, and WLE below .70 by
3.91 points.

Three targeted checks sharpen the picture. Excluding the 4 numerically
degenerate units (WLE below $-1$) moves the true WLE sub-.80 and sub-.50
shares to 51.1\% and 10.2\%, and excluding all 16 further negative units in
the deconvolution sample leaves 49.6\% and 8.7\%. That those units are
genuine now has a basis that an earlier version of this appendix lacked.
Before item keys were aligned (\cref{app:keyalign}) there were 48
negative-WLE units, and we described them as data rather than as artifacts.
Twenty-seven of them were artifacts of mixed keying, so that description was
wrong, and what survives alignment is a smaller and better-founded set.
Inflating the WLE sampling variances by the empirically calibrated factor
1.12 (\cref{app:ses}), or even by 2, leaves the true shares unchanged to one
decimal, at 51.8\% and 11.0\%. Finally, the parametric tail can be checked
against the observed finite-corpus shares, which are assumption-light but
noisy, and against the explicitly model-based diagnostics below. These checks
support the qualitative headline without implying uniform one-point
agreement.

\subsection{Varying versus random coefficients}\label{app:deconv-target}

\Citet{bonett_varying_2010} distinguishes three meta-analytic models that are
usually collapsed into two. A constant-coefficient model assumes one common
parameter. A varying-coefficient model allows each study its own true value
and targets their mean over the studies actually included. A
random-coefficient model, ordinarily called random effects, treats the
included studies as a random sample from a superpopulation and targets that
superpopulation's mean and variance. Bonett's objection to the third is not
that heterogeneity is implausible but that the sampling assumption usually
is, since studies are selected and not drawn, and a superpopulation must then
be described in clear and useful detail or the interval loses its
interpretation.

This objection is relevant here, because we cite Bonett for the $L$ scale
while departing from him on the model. Our corpus attempts every unique
eligible unit in a frozen warehouse snapshot and retains 889 analyzable
administrations among 930 attempts. There is no probability sampling step,
and a finite-corpus target would concern the 889 unknown unit coefficients.
We nevertheless fit a random-coefficient
model for a reason that concerns the estimand rather than the population.
Deconvolution needs a mixing distribution, because stripping estimation noise
from observed dispersion requires a model for what is being stripped toward.

We therefore report complementary summaries. The random-coefficient reading
is the normal-tail formula above. The empirical share of shrunken unit-level
predictions is a Gaussian empirical-Bayes diagnostic: it summarizes the 889
units but still inherits Gaussian exchangeability and shrinkage, and is
neither Bonett's varying-coefficient target nor an assumption-free estimator.
At $c = .80$ it differs from the normal tail by 1.6 points for EAP,
30.3\% against 28.7\%, and by 2.1 points for WLE, 51.8\% against 49.7\%.
Bonett's further concern, that random-coefficient methods degrade when the
superpopulation is not normal, is addressed directly by
\cref{tab:normality}, which finds the $L$ distribution close to Gaussian in
shape. Where the main text intends a superpopulation reading, it should be
understood as instruments of the kind that reach a public warehouse, which is
also the selection frame discussed in the limitations.

\subsection{Sampling dependence within studies}\label{app:deconv-rho}

Units within a study frequently share respondents, for instance subscales of
one battery, so their sampling errors are correlated. Our deconvolution
therefore uses the same constant-sampling-correlation working model as the
meta-regression (Section~\ref{sec:metareg-m}), with $\rho = 0.6$ imputed via
\pkg{metafor}'s \code{vcalc}, rather than treating the $v_{si}$ as
independent. Because the sampling variances are two orders of magnitude
smaller than the heterogeneity components, with median $v_{si} \approx
.0025$ against $\hat\tau^2_{\mathrm{study}} \approx .20$, the choice is
nearly inert. Moving from independence to $\rho = 0.6$ shifts
$\hat\tau^2_{\mathrm{study}}$ by less than .003 and no reported true share
by more than 0.1 percentage points. We use the correlated working model
regardless, so that every model in the paper makes one dependence assumption
rather than two.

\subsection{Assumptions and sensitivity}

The parametric assumption is normal mixing on $L$, and \cref{tab:normality}
diagnoses departures in the observed transformed estimates. In shape the observed $L$ distributions are close to
normal, with skewness $-0.16$ for EAP and $-0.11$ for WLE and kurtosis 3.16
and 3.79, and the Shapiro--Wilk test does not reject for EAP ($W = .997$, $p
= .14$). It does reject for WLE ($p < .001$); at $n = 865$ small departures
are detectable, but the rejection and the heavier WLE lower tail remain
relevant to the sensitivity comparison below.

A more consequential check is the comparison of computations. The last
block of \cref{tab:normality} compares the parametric normal-tail formula, the
observed finite-corpus share, and the empirical share of Gaussian BLUPs. The
observed share is noisy but assumption-light; the BLUP share is explicitly
Gaussian empirical Bayes. At $c = .80$ the
three agree within 2.1 percentage points, giving 30.3\%, 29.1\%, and 28.7\%
for EAP and 51.8\%, 49.7\%, and 49.7\% for WLE. The largest
normal-versus-observed discrepancy is 3.91 points, for WLE below .70. The
small noise share explains the broad agreement but does not force every tail
method to return the same number. We use the parametric model for its
closed-form thresholds and variance decomposition and label the BLUP results
as a Gaussian empirical-Bayes diagnostic.

\input{floats/tabS_normality}

\subsection{A flexible non-Gaussian deconvolution}\label{app:gmodel}

The comparisons above hold the estimator fixed and vary the assumption. A
stronger check varies the method. We re-estimated the mixing distribution
with the $g$-modeling approach of \textcite{efron_deconvolution_2016}, which
represents $g$ as a low-dimensional exponential family over a discrete grid,
here natural splines with five degrees of freedom estimated by penalized
maximum likelihood, rather than assuming normality. Because our sampling
variances are heteroscedastic, we build the likelihood matrix directly,
$P_{ij} = \phi(\hat L_i; \theta_j, \sqrt{v_i})$, and integrate the fitted $g$
above $\ln(1-c)$.

The baseline fit uses a 61-point grid, five natural-spline degrees of
freedom, and penalty 1. We also vary the grid over 41, 61, 101, and 201
points and cross spline degrees of freedom 3, 5, and 7 with penalties 0, 1,
and 5. The fitted $g$ is close to normal but not exactly normal, with skewness
$-0.28$ and kurtosis 3.67 for EAP and $-0.26$ and 4.46 for WLE, against 0 and
3 under the fitted normal, and with means and standard deviations that agree
with the REML fit to three decimals. The consequences are small and
concentrated where the extra kurtosis matters:

\input{floats/tabS_gmodel}

At the headline threshold the two methods agree to 0.1 points for EAP and
1.3 points for WLE. They diverge materially in one place only, the WLE share
below .70, where the normal model reports 30.3\% against the flexible spline
26.0\%. Since the $g$-model tracks the observed share of 26.4\% closely
there and the normal model does not, the gap of 4.3 points should be read as
the normal tail being slightly too heavy rather than as disagreement about
the data. We report the normal-based figures throughout for their closed
form and their variance decomposition, and we flag the WLE sub-.70 share as
the one number in the paper that is sensitive to the mixing assumption by
more than two points.

Our known-variance assumption treats the bootstrap $v_{si}$ as exact, and
\cref{app:ses} quantifies the one direction of concern, the
calibration-fixed WLE variances, showing that the conclusions survive
doubling them. Finally, the threshold statistics sit near the center of the
fitted distributions ($|z| < 0.2$ at $c = .80$), where they are insensitive to modest misestimation of $\hat\tau$. Shrinking
$\hat\tau$ by 10\% moves the WLE true sub-.80 share by less than one
percentage point.

%% file: floats/tabS_deconv.tex
\begin{table}[H]
\centering
\caption{Deconvolution of the observed reliability distribution: three-level random-effects fit on $L=\ln(1-r)$ and implied true shares below conventional thresholds.}
\label{tab:deconv}
\footnotesize
\begin{tabular}{lrrrr}
\toprule
 & \multicolumn{2}{c}{EAP (empirical)} & \multicolumn{2}{c}{WLE (separation)} \\
\cmidrule(lr){2-3}\cmidrule(lr){4-5}
 & Observed & True & Observed & True \\
\midrule
Units in model & \multicolumn{2}{c}{877} & \multicolumn{2}{c}{865} \\
Pooled $r$ & \multicolumn{2}{c}{0.862} & \multicolumn{2}{c}{0.794} \\
$\tau^2$ total (study / unit) & \multicolumn{2}{c}{0.519 (0.201 / 0.318)} & \multicolumn{2}{c}{0.522 (0.181 / 0.342)} \\
Noise share, moment ratio & \multicolumn{2}{c}{1.00\%} & \multicolumn{2}{c}{0.77\%} \\
Noise share, model ratio & \multicolumn{2}{c}{0.95\%} & \multicolumn{2}{c}{0.76\%} \\
\addlinespace
Share $r<.90$ & 68.6\% & 67.2\% & 86.7\% & 84.2\% \\
Share $r<.80$ & 29.1\% & 30.3\% & 49.7\% & 51.8\% \\
Share $r<.70$ & 13.1\% & 14.0\% & 26.4\% & 30.3\% \\
Share $r<.50$ & 3.3\% & 3.7\% & 10.1\% & 11.1\% \\
\bottomrule
\end{tabular}
\end{table}

%% file: floats/tabS_normality.tex
\begin{table}[H]
\centering
\caption{Diagnostics for the normal mixing assumption on $L=\ln(1-r)$. Top block: shape of the observed $L$ distribution. Bottom block contrasts three distinct quantities: the fitted Gaussian random-effects tail, the observed noisy share, and the empirical share of Gaussian-model posterior means (BLUPs). The latter is a shrinkage diagnostic, not an assumption-free estimate of the latent distribution.}
\label{tab:normality}
\footnotesize
\begin{tabular}{lrrrrrr}
\toprule
 & \multicolumn{3}{c}{EAP (empirical)} & \multicolumn{3}{c}{WLE (separation)} \\
\cmidrule(lr){2-4}\cmidrule(lr){5-7}
\midrule
Units in model & \multicolumn{3}{c}{877} & \multicolumn{3}{c}{865} \\
Skewness & \multicolumn{3}{c}{-0.16} & \multicolumn{3}{c}{-0.11} \\
Kurtosis & \multicolumn{3}{c}{3.16} & \multicolumn{3}{c}{3.79} \\
Shapiro--Wilk $W$ ($p$) & \multicolumn{3}{c}{0.9972 (0.14)} & \multicolumn{3}{c}{0.9915 (7.1e-05)} \\
\addlinespace
\multicolumn{7}{l}{\emph{True share below $c$, computed three ways}} \\
 & Normal & Observed & BLUP & Normal & Observed & BLUP \\
\cmidrule(lr){2-4}\cmidrule(lr){5-7}
\quad $c = .90$ & 67.2\% & 68.6\% & 68.6\% & 84.2\% & 86.7\% & 86.9\% \\
\quad $c = .80$ & 30.3\% & 29.1\% & 28.7\% & 51.8\% & 49.7\% & 49.7\% \\
\quad $c = .70$ & 14.0\% & 13.1\% & 12.7\% & 30.3\% & 26.4\% & 26.2\% \\
\quad $c = .50$ & 3.7\% & 3.3\% & 3.3\% & 11.1\% & 10.1\% & 10.1\% \\
\bottomrule
\end{tabular}
\end{table}

%% file: floats/tabS_gmodel.tex
\begin{table}[H]
\centering
\caption{Flexible spline $g$-model deconvolution. The displayed fit uses a 61-point grid, five natural-spline degrees of freedom, penalty 1, and a heteroscedastic Gaussian observation kernel. It is flexible but not assumption-free. Across 41--201 grid points, the sub-.80 estimates range 28.12--30.15\% (EAP) and 50.49--53.06\% (WLE); across df/penalty settings they range 29.48--31.45\% and 50.05--52.09\%.}
\label{tab:gmodel}
\footnotesize
\begin{tabular}{lrrrrrr}
\toprule
 & \multicolumn{3}{c}{EAP (empirical)} & \multicolumn{3}{c}{WLE (separation)} \\
\cmidrule(lr){2-4}\cmidrule(lr){5-7}
Threshold & Normal & $g$-model & Observed & Normal & $g$-model & Observed \\
\midrule
$c = .90$ & 67.2\% & 68.7\% & 68.6\% & 84.2\% & 86.9\% & 86.7\% \\
$c = .80$ & 30.3\% & 30.2\% & 29.1\% & 51.8\% & 50.5\% & 49.7\% \\
$c = .70$ & 14.0\% & 13.5\% & 13.1\% & 30.3\% & 26.0\% & 26.4\% \\
$c = .50$ & 3.7\% & 3.4\% & 3.3\% & 11.1\% & 10.3\% & 10.1\% \\
\bottomrule
\end{tabular}
\end{table}

%% file: appendices/appendix_F.tex
\section{Meta-regression inference, sensitivity, and covariate
audit}\label{app:metareg}

\subsection{The CHE working model and CR2 inference}

Units within a study frequently share respondents, typically as subscales of
one battery administered to one sample, so their sampling errors are
correlated with unknown correlation. Following the
correlated-and-hierarchical-effects approach, we impute a working
within-study correlation $\rho = 0.6$ into a block-diagonal sampling
covariance matrix with blocks by study, retain the two random-effect layers
of the deconvolution model, and estimate by restricted maximum likelihood.
All reported inference uses CR2 cluster-robust variance estimation with
Satterthwaite degrees of freedom, clustering on study
\parencite{hedges_robust_2010, satterthwaite_approximate_1946,
pustejovsky_small_2018, pustejovsky_clubsandwich_2026}. With independent study
clusters, adequate cluster count, and the usual CR2 regularity conditions,
this makes inference less sensitive to the imputed $\rho$ and to working-model
misspecification; it does not remove every dependence assumption. In
\cref{tab:rho} we verify practical stability empirically for the displayed
focal coefficients: each moves by at most .001 across
$\rho \in \{0.3, 0.6, 0.9\}$.

\input{floats/tabS_rho}

\subsection{Moderator tiers}

Moderators were fixed in advance in three tiers. A confirmatory tier,
reported in \cref{tab:metareg} of the main text, contains only design
descriptors knowable before any model is fitted, namely log items, log
persons, response categories, construct type, and longitudinal origin. A
mediator tier contains quantities that are mechanically posterior to the
calibration, namely the unidimensionality ratio $\lambda_1/\lambda_2$, the
ratio of the first two eigenvalues of the item correlation matrix, and the
model-family assignment, dichotomous Rasch versus polytomous GPCM. These sit
on the causal path from design to reliability and would over-control the
confirmatory contrasts. An exploratory tier contains the remaining
harvestable descriptors: sample frame, a sparse-design flag, provenance
grade, and administration policy. Only the second and third tiers enter the
audit below.

\subsection{Explained heterogeneity}

Relative to the intercept-only model with the same working covariance, the
confirmatory moderators explain 35\% of total heterogeneity under EAP and
19\% under WLE. Most of the heterogeneity, and especially its within-study
component, remains unexplained by design descriptors, which is the
quantitative form of the recommendation that reliability be measured rather
than predicted.

\subsection{The implementation/specification-gap test}

\input{floats/tabS_delta}

Because the two columns of \cref{tab:metareg} in the main text are fitted to
the same units, comparing their significance patterns is not a test of
implementation/specification-gap moderation
\parencite{gelman_difference_2006}. In \cref{tab:delta} we report the direct
test, in which within-unit differences $\Delta_i =
L^{\mathrm{WLE}}_i - L^{\mathrm{EAP}}_i$ are regressed on the confirmatory
moderators with CR2 study-clustered errors ($n=875$ units in 430 study
clusters). Differencing absorbs additive random effects shared exactly by the
two outputs, but it does not eliminate package, calibration, slope-model, or
potentially unequal sampling-error components. The regression therefore
describes moderation of the implemented WLE--EAP gap, not estimand alone. As
in the main text, the construct
factor is jointly significant (HTZ $F(8, 25.3) = 2.72$, $p = .026$), and the
personality, behavioral, and opinion--attitude contrasts survive Holm
correction while the affective contrast does not. Two further design results
follow from the same regression. First, the implemented gap grows with
response categories, at $+.047$ per category with $p < .001$; this can reflect
shrinkage together with the package and slope-model contrast in polytomous
instruments. Second, the
log-persons slope, which is null within the small-sample stratum for both
estimands ($b = .008$ for EAP and $.020$ for WLE, both non-significant),
identifies the pooled positive coefficient of \cref{tab:metareg} as a
between-stratum composition contrast rather than as a within-stratum
gradient. Note that coefficients in this subsection come from the unweighted
CR2 refit and may differ from the REML fits of \cref{tab:metareg} in the
third decimal.

\subsection{The residual-covariate audit}

\input{floats/tabS_sens}

In \cref{tab:audit} we report the audit, in which every exploratory and
mediator covariate is added simultaneously to the confirmatory model. Three
results stand out. First, the null results that matter are clean: sample
frame, whether clinical, general, internet, or targeted, moves nothing, and
the provisionally ingested collection is indistinguishable from the
canonical one ($\hat\beta = -.013$ under EAP), which supports pooling the
corpus. Second,
sparse planned-missing designs carry a genuine penalty under EAP ($+.256$,
$p < .001$). Third, the two mediators dominate. The unidimensionality ratio
is the strongest single correlate in the entire analysis ($-.465$ for EAP and
$-.483$ for WLE per log unit, both $p < .0001$), and the model-family
indicator is large ($+.686$ for EAP), reflecting the confound between
dichotomous formats and the corpus's cognitive tests.

\begin{figure}[H]
\centering
\includegraphics[width=0.9\textwidth]{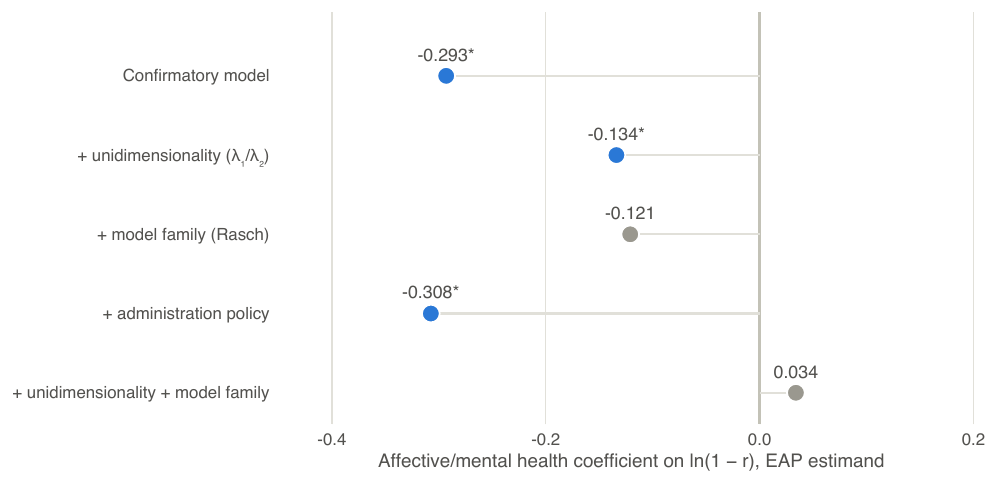}
\caption{Mediation decomposition of the EAP affective-advantage coefficient.
Adding the unidimensionality ratio or the model family individually halves
the contrast, and adding both eliminates it. Administration policy, included
as a placebo control, leaves it untouched. Filled points mark $p<.05$.}
\label{fig:mediation}
\end{figure}

This audit reframes the construct-type results of the main text. The EAP
affective advantage of $-.293$ behaves like a composition effect, being
halved by conditioning on unidimensionality ($-.134$), halved by model
family ($-.121$, non-significant), untouched by the placebo ($-.307$), and
eliminated by the two jointly ($+.034$, non-significant;
\cref{fig:mediation}). Affective screeners are thus more reliable because
they are strongly one-dimensional polytomous instruments, and not because of
their content domain as such. By contrast, the WLE personality penalty is
not explained away and in fact strengthens slightly under the full audit
block, from $+.248$ to $+.300$ with $p = .008$. That pattern is consistent
with an information-curvature mechanism, in which narrow measurement bands
are weighted heavily by the unbounded $\wbar$ integrand, rather than with a
composition story. We report the confirmatory total-effect estimates in the
main text and the decomposition here, because conditioning on mediators
would answer a different question.

%% file: floats/tabS_rho.tex
\begin{table}[H]
\centering
\caption{Sensitivity of the CHE meta-regression to the imputed within-study sampling correlation $\rho$. Coefficients on $\ln(1-r)$; reference construct Cognitive/educational. Under the stated CR2 regularity conditions, the displayed focal coefficients are practically stable: each moves by at most .001 across the working correlations.}
\label{tab:rho}
\footnotesize
\begin{tabular}{llrrrrr}
\toprule
Estimand & $\rho$ & log(items) & log(persons) & Categories & Affective & Personality \\
\midrule
EAP & 0.3 & $-0.385$ & $+0.053$ & $-0.127$ & $-0.293$ & $-0.149$ \\
EAP & 0.6 & $-0.385$ & $+0.053$ & $-0.127$ & $-0.293$ & $-0.149$ \\
EAP & 0.9 & $-0.385$ & $+0.053$ & $-0.127$ & $-0.293$ & $-0.149$ \\
WLE & 0.3 & $-0.340$ & $+0.054$ & $-0.083$ & $-0.204$ & $+0.248$ \\
WLE & 0.6 & $-0.341$ & $+0.054$ & $-0.083$ & $-0.204$ & $+0.248$ \\
WLE & 0.9 & $-0.341$ & $+0.055$ & $-0.083$ & $-0.204$ & $+0.248$ \\
\bottomrule
\end{tabular}
\end{table}

%% file: floats/tabS_delta.tex
\begin{table}[H]
\centering
\caption{Implementation/specification-gap moderation: within-unit differences $\Delta_i = L_i^{\mathrm{WLE}} - L_i^{\mathrm{EAP}}$ regressed on the confirmatory moderators, with CR2 study-clustered standard errors ($n=875$ units, 430 clusters). Because the WLE and EAP branches also differ in software, calibration, and slope model for polytomous data, these coefficients do not isolate a pure estimand effect. Construct contrasts carry Holm-adjusted $p$ values (family of 8). Omnibus construct test: HTZ $F(8, 25.3) = 2.72$, $p = 0.026$.}
\label{tab:delta}
\footnotesize
\begin{tabular}{lcrrr}
\toprule
Term & $\hat\beta$ (SE) & df & $p$ & $p_{\mathrm{Holm}}$ \\
\midrule
log(items) & $+0.042$ (0.018) & 76 & 0.0197 & --- \\
log(persons) & $-0.005$ (0.007) & 95 & 0.4705 & --- \\
Response categories & $+0.047$ (0.008) & 46 & 0.0000 & --- \\
Longitudinal & $+0.069$ (0.037) & 24 & 0.0735 & --- \\
\addlinespace
\multicolumn{5}{l}{\itshape Construct type (ref.\ Cognitive/educational)} \\
\quad Affective/mental health & $+0.074$ (0.042) & 86 & 0.0825 & 0.3298 \\
\quad Behavioral & $+0.186$ (0.063) & 37 & 0.0050 & 0.0302 \\
\quad Developmental & $+0.036$ (0.046) & 2 & 0.5050 & 1.0000 \\
\quad Opinion/attitude & $+0.133$ (0.043) & 62 & 0.0030 & 0.0209 \\
\quad Other & $+0.092$ (0.140) & 4 & 0.5482 & 1.0000 \\
\quad Personality & $+0.324$ (0.085) & 38 & 0.0005 & 0.0039 \\
\quad Physical health/functioning & $+0.048$ (0.062) & 13 & 0.4582 & 1.0000 \\
\quad Unclassified & $+0.108$ (0.044) & 79 & 0.0158 & 0.0792 \\
\bottomrule
\end{tabular}
\end{table}

%% file: floats/tabS_sens.tex
\begin{table}[H]
\centering
\caption{Residual-covariate audit: confirmatory coefficients before and after adding every remaining harvestable dataset property (sample frame, sparse design, unidimensionality, provenance grade, model family, administration policy). Stars: $^{*}p<.05$, $^{**}p<.01$, $^{***}p<.001$ (CR2). Unidimensionality and model family are mechanically posterior to the fit (mediators), which is why they are excluded from the confirmatory model.}
\label{tab:audit}
\footnotesize
\setlength{\tabcolsep}{5pt}
\begin{tabular}{lcccc}
\toprule
 & \multicolumn{2}{c}{EAP} & \multicolumn{2}{c}{WLE} \\
\cmidrule(lr){2-3}\cmidrule(lr){4-5}
Term & Confirmatory & +Audit block & Confirmatory & +Audit block \\
\midrule
log(items) & $-0.385$$^{***}$ & $-0.526$$^{***}$ & $-0.341$$^{***}$ & $-0.476$$^{***}$ \\
log(persons) & $+0.053$$^{***}$ & $+0.049$$^{***}$ & $+0.054$$^{***}$ & $+0.052$$^{***}$ \\
Response categories & $-0.127$$^{***}$ & $-0.060$$^{***}$ & $-0.083$$^{***}$ & $-0.032$$^{*}$ \\
Longitudinal & $+0.029$ & $+0.060$ & $+0.114$ & $+0.189$ \\
Affective/mental health & $-0.293$$^{***}$ & $+0.014$ & $-0.204$$^{*}$ & $+0.072$ \\
Personality & $-0.149$ & $-0.012$ & $+0.248$$^{*}$ & $+0.301$$^{**}$ \\
Sample: clinical & --- & $+0.086$ & --- & $-0.015$ \\
Sample: general & --- & $+0.010$ & --- & $+0.008$ \\
Sample: internet & --- & $-0.013$ & --- & $+0.042$ \\
Sample: targeted/other & --- & $-0.017$ & --- & $+0.076$ \\
Sample: unknown & --- & $-0.147$ & --- & $-0.153$ \\
Sparse design (density $<0.9$) & --- & $+0.256$$^{***}$ & --- & $+0.210$$^{**}$ \\
log($\lambda_1/\lambda_2$) & --- & $-0.465$$^{***}$ & --- & $-0.483$$^{***}$ \\
Provisional provenance (grade B) & --- & $-0.013$ & --- & $-0.055$ \\
Dichotomous/Rasch family & --- & $+0.686$$^{***}$ & --- & $+0.565$$^{***}$ \\
Admin.: single & --- & $-0.214$$^{**}$ & --- & $-0.205$$^{*}$ \\
Admin.: wave selected & --- & $-0.419$$^{**}$ & --- & $-0.492$$^{**}$ \\
\bottomrule
\end{tabular}
\end{table}

%% file: appendices/appendix_G.tex
\section{Additional results}\label{app:additional}

\subsection{Reliability by construct type}

\begin{figure}[H]
\centering
\includegraphics[width=0.95\textwidth]{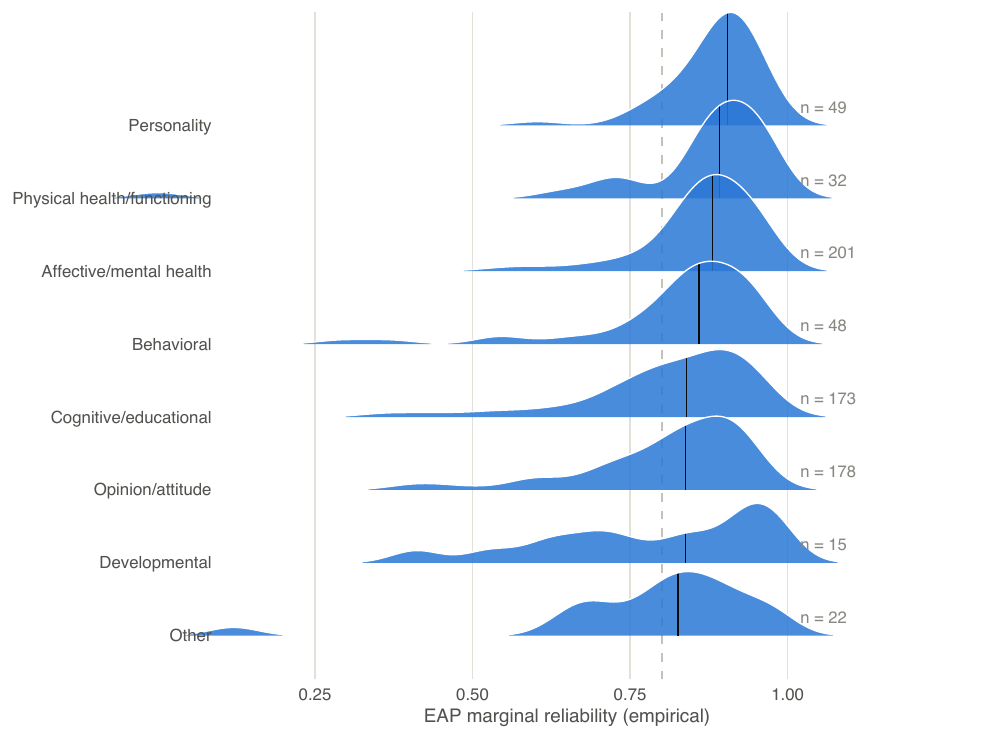}
\caption{EAP empirical reliability by construct type, using warehouse
metadata and omitting unclassified units. Ridgelines are kernel densities
with the median marked, ordered by median. Every construct family carries a
substantial sub-.80 tail, and the apparent advantage of affective and
mental-health scales is decomposed in \cref{fig:mediation}.}
\label{fig:construct}
\end{figure}

In \cref{fig:construct} we display the distribution within each construct
family under the EAP estimand. Two features go beyond the meta-regression
coefficients. First, within-family dispersion is much larger than
between-family differences in location, which is consistent with the variance
decomposition ($\tau^2_{\mathrm{unit}} > \tau^2_{\mathrm{study}}$). Second,
no family is exempt from the low tail. Even personality measurement, the top
family by EAP median, places roughly one unit in eight below .80, and the
ordering of families is itself estimand-dependent, since under WLE the
personality family drops toward the bottom (\cref{tab:metareg}).

\subsection{Full specification-curve table}

\input{floats/tabS_spec}

\Cref{tab:specfull} tabulates \cref{fig:spec}, adding confidence intervals
and the common-subset estimates. Maximal-sample and common-subset columns
agree to within .010 in pooled reliability for every estimator, the largest
movement being WLE separation from .794 to .803, and the six estimators
arrange in the same order in both columns. Coverage composition is therefore
ruled out as an explanation for the estimator spread.

\subsection{Composition by size stratum}

\input{floats/tabS_strata}

\Cref{tab:strata} tabulates \cref{fig:composition}. Under both co-primary estimands the second completion wave, the 101 largest
datasets of which 94 and 100 carry standard errors under EAP and WLE, is
less reliable than the first wave, with pooled $r$ of .848 against .865
under EAP and .773 against .799 under WLE, so completing the corpus raised
the full-corpus sub-.80 shares. Composition explains the direction, since
the warehouse's largest datasets are internet panels and national surveys
carrying short embedded scales, such as values and attitude modules, rather
than professionally constructed achievement tests. This composition is also
the source of the positive log-persons coefficient in the meta-regression.

%% file: floats/tabS_spec.tex
\begin{table}[H]
\centering
\caption{The estimator-sensitivity analysis in full: pooled reliability and true sub-threshold shares for each estimator, on its maximal available sample and on the common subset of 623 units carrying standard errors for all six estimators (per-estimator domain filters leave 621--623 effective, as the $n$ column shows).}
\label{tab:specfull}
\footnotesize
\begin{tabular}{llrrcrr}
\toprule
Set & Estimator & $n$ & Pooled $r$ & 95\% CI & True $<.80$ & True $<.70$ \\
\midrule
Maximal $n$ & WLE (separation) & 865 & 0.794 & [0.780, 0.806] & 51.8\% & 30.3\% \\
Maximal $n$ & Cronbach's $\alpha$ & 647 & 0.843 & [0.831, 0.854] & 37.1\% & 18.8\% \\
Maximal $n$ & Guttman $\lambda_2$ & 647 & 0.854 & [0.843, 0.864] & 33.2\% & 15.9\% \\
Maximal $n$ & TRC & 807 & 0.860 & [0.849, 0.871] & 33.2\% & 17.6\% \\
Maximal $n$ & EAP (empirical) & 877 & 0.862 & [0.853, 0.870] & 30.3\% & 14.0\% \\
Maximal $n$ & MRC & 805 & 0.868 & [0.860, 0.875] & 26.0\% & 10.2\% \\
Common subset & WLE (separation) & 621 & 0.803 & [0.789, 0.816] & 49.0\% & 26.2\% \\
Common subset & Cronbach's $\alpha$ & 623 & 0.844 & [0.832, 0.855] & 36.6\% & 18.3\% \\
Common subset & Guttman $\lambda_2$ & 623 & 0.854 & [0.843, 0.864] & 32.9\% & 15.6\% \\
Common subset & TRC & 623 & 0.862 & [0.850, 0.874] & 32.3\% & 16.9\% \\
Common subset & EAP (empirical) & 623 & 0.862 & [0.852, 0.871] & 28.8\% & 12.1\% \\
Common subset & MRC & 623 & 0.865 & [0.855, 0.874] & 27.1\% & 10.7\% \\
\bottomrule
\end{tabular}
\end{table}

%% file: floats/tabS_strata.tex
\begin{table}[H]
\centering
\caption{Observed and deconvolved true shares by completion wave. Small/Medium denotes the 788 first-wave datasets and Large/XL the 101 second-wave datasets (all 75 XL-tier plus the 26 largest medium-tier; Table 1); the $n$ column counts the datasets in each wave carrying a usable standard error for the given estimator, 94 and 100 of the second wave under EAP and WLE. The second wave is less reliable than the first, so completing the corpus raised the full-corpus sub-.80 shares.}
\label{tab:strata}
\footnotesize
\resizebox{\textwidth}{!}{%
\begin{tabular}{llrrrrrr}
\toprule
Estimand & Subset & $n$ & Pooled $r$ & Obs.\ $<.80$ & True $<.80$ & True $<.70$ & True $<.50$ \\
\midrule
EAP (empirical) & Small/Medium & 783 & 0.865 & 27.6\% & 28.9\% & 13.0\% & 3.2\% \\
EAP (empirical) & Large/XL & 94 & 0.848 & 41.5\% & 36.9\% & 20.5\% & 7.5\% \\
EAP (empirical) & Full corpus & 877 & 0.862 & 29.1\% & 30.3\% & 14.0\% & 3.7\% \\
WLE (separation) & Small/Medium & 765 & 0.799 & 47.2\% & 50.3\% & 29.0\% & 10.4\% \\
WLE (separation) & Large/XL & 100 & 0.773 & 69.0\% & 56.6\% & 35.7\% & 14.9\% \\
WLE (separation) & Full corpus & 865 & 0.794 & 49.7\% & 51.8\% & 30.3\% & 11.1\% \\
\bottomrule
\end{tabular}%
}
\end{table}

%% file: references.bib
@article{domingue_introduction_2025,
	author = {Domingue, Benjamin W. and Braginsky, Mika and Caffrey-Maffei, Lucy and Gilbert, Joshua B. and Kanopka, Klint and Kapoor, Radhika and Lee, Hansol and Liu, Yiqing and Nadela, Savira and Pan, Guanzhong and Zhang, Lijin and Zhang, Susu and Frank, Michael C.},
	title = {An introduction to the item response warehouse ({IRW}): {A} resource for enhancing data usage in psychometrics},
	journal = {Behavior Research Methods},
	year = {2025},
	volume = {57},
	number = {10},
	pages = {276},
	doi = {10.3758/s13428-025-02796-y},
	issn = {1554-3528}
}

@misc{lee_reliability_2025,
	author = {Lee, JoonHo},
	title = {Reliability-targeted simulation of item response data: Solving the inverse design problem},
	year = {2025},
	howpublished = {arXiv preprint},
	note = {arXiv:2512.16012}
}

@book{lord_novick_1968,
	author = {Lord, Frederic M. and Novick, Melvin R.},
	title = {Statistical theories of mental test scores},
	year = {1968},
	publisher = {Addison-Wesley},
	address = {Reading, MA}
}

@article{cronbach_coefficient_1951,
	author = {Cronbach, Lee J.},
	title = {Coefficient alpha and the internal structure of tests},
	journal = {Psychometrika},
	year = {1951},
	volume = {16},
	number = {3},
	pages = {297--334},
	doi = {10.1007/BF02310555}
}

@article{guttman_basis_1945,
	author = {Guttman, Louis},
	title = {A basis for analyzing test-retest reliability},
	journal = {Psychometrika},
	year = {1945},
	volume = {10},
	number = {4},
	pages = {255--282},
	doi = {10.1007/BF02288892}
}

@book{mcdonald_test_1999,
	author = {McDonald, Roderick P.},
	title = {Test theory: {A} unified treatment},
	year = {1999},
	publisher = {Lawrence Erlbaum},
	address = {Mahwah, NJ}
}

@article{sijtsma_use_2009,
	author = {Sijtsma, Klaas},
	title = {On the use, the misuse, and the very limited usefulness of {Cronbach's} alpha},
	journal = {Psychometrika},
	year = {2009},
	volume = {74},
	number = {1},
	pages = {107--120},
	doi = {10.1007/s11336-008-9101-0}
}

@article{mcneish_thanks_2018,
	author = {McNeish, Daniel},
	title = {Thanks coefficient alpha, we'll take it from here},
	journal = {Psychological Methods},
	year = {2018},
	volume = {23},
	number = {3},
	pages = {412--433},
	doi = {10.1037/met0000144}
}

@article{revelle_reliability_2019,
	author = {Revelle, William and Condon, David M.},
	title = {Reliability from $\alpha$ to $\omega$: {A} tutorial},
	journal = {Psychological Assessment},
	year = {2019},
	volume = {31},
	number = {12},
	pages = {1395--1411},
	doi = {10.1037/pas0000754},
	issn = {1939-134X}
}

@book{nunnally_psychometric_1994,
	author = {Nunnally, Jum C. and Bernstein, Ira H.},
	title = {Psychometric theory},
	edition = {3rd},
	year = {1994},
	publisher = {McGraw-Hill},
	address = {New York}
}

@article{cronbach_signalnoise_1964,
	author = {Cronbach, Lee J. and Gleser, Goldine C.},
	title = {The signal/noise ratio in the comparison of reliability coefficients},
	journal = {Educational and Psychological Measurement},
	year = {1964},
	volume = {24},
	number = {3},
	pages = {467--480},
	doi = {10.1177/001316446402400303},
	issn = {0013-1644}
}

@article{warm_weighted_1989,
	author = {Warm, Thomas A.},
	title = {Weighted likelihood estimation of ability in item response theory},
	journal = {Psychometrika},
	year = {1989},
	volume = {54},
	number = {3},
	pages = {427--450},
	doi = {10.1007/BF02294627}
}

@article{bock_adaptive_1982,
	author = {Bock, R. Darrell and Mislevy, Robert J.},
	title = {Adaptive {EAP} estimation of ability in a microcomputer environment},
	journal = {Applied Psychological Measurement},
	year = {1982},
	volume = {6},
	number = {4},
	pages = {431--444},
	doi = {10.1177/014662168200600405},
	issn = {0146-6216}
}

@article{green_technical_1984,
	author = {Green, Bert F. and Bock, R. Darrell and Humphreys, Lloyd G. and Linn, Robert L. and Reckase, Mark D.},
	title = {Technical guidelines for assessing computerized adaptive tests},
	journal = {Journal of Educational Measurement},
	year = {1984},
	volume = {21},
	number = {4},
	pages = {347--360},
	doi = {10.1111/j.1745-3984.1984.tb01039.x}
}

@article{adams_reliability_2005,
	author = {Adams, Raymond J.},
	title = {Reliability as a measurement design effect},
	journal = {Studies in Educational Evaluation},
	year = {2005},
	volume = {31},
	number = {2-3},
	pages = {162--172},
	doi = {10.1016/j.stueduc.2005.05.008},
	issn = {0191-491X}
}

@book{wright_masters_1982,
	author = {Wright, Benjamin D. and Masters, Geofferey N.},
	title = {Rating scale analysis: {Rasch} measurement},
	year = {1982},
	publisher = {MESA Press},
	address = {Chicago}
}

@article{kim_note_2012,
	author = {Kim, Seonghoon},
	title = {A note on the reliability coefficients for item response model-based ability estimates},
	journal = {Psychometrika},
	year = {2012},
	volume = {77},
	number = {1},
	pages = {153--162},
	doi = {10.1007/s11336-011-9238-0},
	issn = {0033-3123}
}

@article{kim_estimation_2010,
	author = {Kim, Seonghoon and Feldt, Leonard S.},
	title = {The estimation of the {IRT} reliability coefficient and its lower and upper bounds, with comparisons to {CTT} reliability statistics},
	journal = {Asia Pacific Education Review},
	year = {2010},
	volume = {11},
	number = {2},
	pages = {179--188},
	doi = {10.1007/s12564-009-9062-8},
	issn = {1876-407X}
}

@article{cheng_comparison_2012,
	author = {Cheng, Ying and Yuan, Ke-Hai and Liu, Cheng},
	title = {Comparison of reliability measures under factor analysis and item response theory},
	journal = {Educational and Psychological Measurement},
	year = {2012},
	volume = {72},
	number = {1},
	pages = {52--67},
	doi = {10.1177/0013164411407315},
	issn = {0013-1644}
}

@inproceedings{zhang_empirical_2006,
	author = {Zhang, Yanwei and Breithaupt, Krista and Tessema, Aster and Chuah, David},
	title = {Empirical vs.\ expected {IRT}-based reliability estimation in computerized multistage testing ({MST})},
	booktitle = {Paper presented at the Annual Meeting of the National Council on Measurement in Education},
	year = {2006},
	address = {San Francisco, CA},
	note = {ERIC ED498534}
}

@article{yang_characterizing_2012,
	author = {Yang, Ji Seung and Hansen, Mark and Cai, Li},
	title = {Characterizing sources of uncertainty in item response theory scale scores},
	journal = {Educational and Psychological Measurement},
	year = {2012},
	volume = {72},
	number = {2},
	pages = {264--290},
	doi = {10.1177/0013164411410056},
	issn = {0013-1644}
}

@article{doran_information_2005,
	author = {Doran, Harold C.},
	title = {The information function for the one-parameter logistic model: Is it reliability?},
	journal = {Educational and Psychological Measurement},
	year = {2005},
	volume = {65},
	number = {5},
	pages = {665--675},
	doi = {10.1177/0013164404272500},
	issn = {0013-1644}
}

@article{soland_how_2024,
	author = {Soland, James and Kuhfeld, Megan and Edwards, Kelly},
	title = {How survey scoring decisions can influence your study's results: {A} trip through the {IRT} looking glass},
	journal = {Psychological Methods},
	year = {2024},
	volume = {29},
	number = {5},
	pages = {1003--1024},
	doi = {10.1037/met0000506},
	issn = {1082-989X}
}

@article{ark_standard_2025,
	author = {{van der Ark}, L. Andries},
	title = {Standard errors for reliability coefficients},
	journal = {Psychometrika},
	year = {2025},
	volume = {90},
	number = {5},
	pages = {1679--1704},
	doi = {10.1017/psy.2025.10050},
	issn = {0033-3123}
}

@article{andersson_large_2018,
	author = {Andersson, Björn and Xin, Tao},
	title = {Large sample confidence intervals for item response theory reliability coefficients},
	journal = {Educational and Psychological Measurement},
	year = {2018},
	volume = {78},
	number = {1},
	pages = {32--45},
	doi = {10.1177/0013164417713570},
	issn = {0013-1644}
}

@misc{andersson_irtreliability_2022,
	author = {Andersson, Björn},
	title = {irtreliability: Item response theory reliability},
	year = {2022},
	howpublished = {R package version 0.1-1},
	doi = {10.32614/CRAN.package.irtreliability}
}

@article{vachahaase_reliability_1998,
	author = {Vacha-Haase, Tammi},
	title = {Reliability generalization: {Exploring} variance in measurement error affecting score reliability across studies},
	journal = {Educational and Psychological Measurement},
	year = {1998},
	volume = {58},
	number = {1},
	pages = {6--20},
	doi = {10.1177/0013164498058001002}
}

@article{thompson_psychometrics_2000,
	author = {Thompson, Bruce and Vacha-Haase, Tammi},
	title = {Psychometrics is datametrics: {The} test is not reliable},
	journal = {Educational and Psychological Measurement},
	year = {2000},
	volume = {60},
	number = {2},
	pages = {174--195},
	doi = {10.1177/0013164400602002}
}

@article{sanchezmeca_recommended_2013,
	author = {S{\'a}nchez-Meca, Julio and L{\'o}pez-L{\'o}pez, Jos{\'e} A. and L{\'o}pez-Pina, Jos{\'e} A.},
	title = {Some recommended statistical analytic practices when reliability generalization studies are conducted},
	journal = {British Journal of Mathematical and Statistical Psychology},
	year = {2013},
	volume = {66},
	number = {3},
	pages = {402--425},
	doi = {10.1111/j.2044-8317.2012.02057.x}
}

@article{flake_measurement_2020,
	author = {Flake, Jessica Kay and Fried, Eiko I.},
	title = {Measurement schmeasurement: {Questionable} measurement practices and how to avoid them},
	journal = {Advances in Methods and Practices in Psychological Science},
	year = {2020},
	volume = {3},
	number = {4},
	pages = {456--465},
	doi = {10.1177/2515245920952393}
}

@article{hedge_reliability_2018,
	author = {Hedge, Craig and Powell, Georgina and Sumner, Petroc},
	title = {The reliability paradox: {Why} robust cognitive tasks do not produce reliable individual differences},
	journal = {Behavior Research Methods},
	year = {2018},
	volume = {50},
	number = {3},
	pages = {1166--1186},
	doi = {10.3758/s13428-017-0935-1}
}

@article{parsons_psychological_2019,
	author = {Parsons, Sam and Kruijt, Anne-Wil and Fox, Elaine},
	title = {Psychological science needs a standard practice of reporting the reliability of cognitive-behavioral measurements},
	journal = {Advances in Methods and Practices in Psychological Science},
	year = {2019},
	volume = {2},
	number = {4},
	pages = {378--395},
	doi = {10.1177/2515245919879695}
}

@article{loken_measurement_2017,
	author = {Loken, Eric and Gelman, Andrew},
	title = {Measurement error and the replication crisis},
	journal = {Science},
	year = {2017},
	volume = {355},
	number = {6325},
	pages = {584--585},
	doi = {10.1126/science.aal3618}
}

@article{rouder_hierarchical_2024,
	author = {Rouder, Jeffrey N. and Mehrvarz, Mahbod},
	title = {Hierarchical-model insights for planning and interpreting individual-difference studies of cognitive abilities},
	journal = {Current Directions in Psychological Science},
	year = {2024},
	volume = {33},
	number = {2},
	pages = {128--135},
	doi = {10.1177/09637214231220923},
	issn = {0963-7214}
}

@article{steegen_increasing_2016,
	author = {Steegen, Sara and Tuerlinckx, Francis and Gelman, Andrew and Vanpaemel, Wolf},
	title = {Increasing transparency through a multiverse analysis},
	journal = {Perspectives on Psychological Science},
	year = {2016},
	volume = {11},
	number = {5},
	pages = {702--712},
	doi = {10.1177/1745691616658637}
}

@article{simonsohn_specification_2020,
	author = {Simonsohn, Uri and Simmons, Joseph P. and Nelson, Leif D.},
	title = {Specification curve analysis},
	journal = {Nature Human Behaviour},
	year = {2020},
	volume = {4},
	number = {11},
	pages = {1208--1214},
	doi = {10.1038/s41562-020-0912-z}
}

@article{bonett_sample_2002,
	author = {Bonett, Douglas G.},
	title = {Sample size requirements for testing and estimating coefficient alpha},
	journal = {Journal of Educational and Behavioral Statistics},
	year = {2002},
	volume = {27},
	number = {4},
	pages = {335--340},
	doi = {10.3102/10769986027004335}
}

@article{bonett_varying_2010,
	author = {Bonett, Douglas G.},
	title = {Varying coefficient meta-analytic methods for alpha reliability},
	journal = {Psychological Methods},
	year = {2010},
	volume = {15},
	number = {4},
	pages = {368--385},
	doi = {10.1037/a0020142}
}

@article{hedges_robust_2010,
	author = {Hedges, Larry V. and Tipton, Elizabeth and Johnson, Matthew C.},
	title = {Robust variance estimation in meta-regression with dependent effect size estimates},
	journal = {Research Synthesis Methods},
	year = {2010},
	volume = {1},
	number = {1},
	pages = {39--65},
	doi = {10.1002/jrsm.5}
}

@article{konstantopoulos_fixed_2011,
	author = {Konstantopoulos, Spyros},
	title = {Fixed effects and variance components estimation in three-level meta-analysis},
	journal = {Research Synthesis Methods},
	year = {2011},
	volume = {2},
	number = {1},
	pages = {61--76},
	doi = {10.1002/jrsm.35}
}

@article{vandennoortgate_threelevel_2013,
	author = {Van den Noortgate, Wim and L{\'o}pez-L{\'o}pez, Jos{\'e} A. and Mar{\'i}n-Mart{\'i}nez, Fulgencio and S{\'a}nchez-Meca, Julio},
	title = {Three-level meta-analysis of dependent effect sizes},
	journal = {Behavior Research Methods},
	year = {2013},
	volume = {45},
	number = {2},
	pages = {576--594},
	doi = {10.3758/s13428-012-0261-6}
}

@article{pustejovsky_small_2018,
	author = {Pustejovsky, James E. and Tipton, Elizabeth},
	title = {Small-sample methods for cluster-robust variance estimation and hypothesis testing in fixed effects models},
	journal = {Journal of Business \& Economic Statistics},
	year = {2018},
	volume = {36},
	number = {4},
	pages = {672--683},
	doi = {10.1080/07350015.2016.1247004}
}

@article{viechtbauer_conducting_2010,
	author = {Viechtbauer, Wolfgang},
	title = {Conducting meta-analyses in {R} with the metafor package},
	journal = {Journal of Statistical Software},
	year = {2010},
	volume = {36},
	number = {3},
	pages = {1--48},
	doi = {10.18637/jss.v036.i03}
}

@article{efron_deconvolution_2016,
	author = {Efron, Bradley},
	title = {Empirical {Bayes} deconvolution estimates},
	journal = {Biometrika},
	year = {2016},
	volume = {103},
	number = {1},
	pages = {1--20},
	doi = {10.1093/biomet/asv068}
}

@book{rasch_probabilistic_1960,
	author = {Rasch, Georg},
	title = {Probabilistic models for some intelligence and attainment tests},
	year = {1960},
	publisher = {Danish Institute for Educational Research},
	address = {Copenhagen}
}

@article{masters_partial_1982,
	author = {Masters, Geoff N.},
	title = {A {Rasch} model for partial credit scoring},
	journal = {Psychometrika},
	year = {1982},
	volume = {47},
	number = {2},
	pages = {149--174},
	doi = {10.1007/BF02296272}
}

@article{muraki_generalized_1992,
	author = {Muraki, Eiji},
	title = {A generalized partial credit model: {Application} of an {EM} algorithm},
	journal = {Applied Psychological Measurement},
	year = {1992},
	volume = {16},
	number = {2},
	pages = {159--176},
	doi = {10.1177/014662169201600206}
}

@article{chalmers_mirt_2012,
	author = {Chalmers, R. Philip},
	title = {mirt: {A} multidimensional item response theory package for the {R} environment},
	journal = {Journal of Statistical Software},
	year = {2012},
	volume = {48},
	number = {6},
	pages = {1--29},
	doi = {10.18637/jss.v048.i06}
}

@misc{robitzsch_tam_2026,
	author = {Robitzsch, Alexander and Kiefer, Thomas and Wu, Margaret},
	title = {{TAM}: Test analysis modules},
	year = {2026},
	howpublished = {R package version 4.3-25},
	doi = {10.32614/CRAN.package.TAM}
}

@misc{pustejovsky_clubsandwich_2026,
	author = {Pustejovsky, James E.},
	title = {clubSandwich: Cluster-robust (sandwich) variance estimators with small-sample corrections},
	year = {2026},
	howpublished = {R package version 0.7.0},
	doi = {10.32614/CRAN.package.clubSandwich}
}

@misc{rcoreteam_2026,
	author = {{R Core Team}},
	title = {R: {A} language and environment for statistical computing},
	year = {2026},
	howpublished = {R Foundation for Statistical Computing, Vienna, Austria. Version 4.6.0},
	url = {https://www.R-project.org/}
}

@article{enders_impact_2004,
	author = {Enders, Craig K.},
	title = {The impact of missing data on sample reliability estimates: {Implications} for reliability reporting practices},
	journal = {Educational and Psychological Measurement},
	year = {2004},
	volume = {64},
	number = {3},
	pages = {419--436},
	doi = {10.1177/0013164403261050},
	issn = {0013-1644}
}

@article{raykov_importance_2023,
	author = {Raykov, Tenko and Anthony, James C. and Menold, Natalja},
	title = {On the importance of coefficient alpha for measurement research: {Loading} equality is not necessary for alpha's utility as a scale reliability index},
	journal = {Educational and Psychological Measurement},
	year = {2023},
	volume = {83},
	number = {4},
	pages = {766--781},
	doi = {10.1177/00131644221104972},
	issn = {0013-1644}
}

@article{kim_nicewander_1993,
	author = {Kim, Jwa K. and Nicewander, W. Alan},
	title = {Ability estimation for conventional tests},
	journal = {Psychometrika},
	year = {1993},
	volume = {58},
	number = {4},
	pages = {587--599},
	doi = {10.1007/BF02294829}
}

@article{culpepper_reliability_2013,
	author = {Culpepper, Steven Andrew},
	title = {The reliability and precision of total scores and {IRT} estimates as a function of polytomous {IRT} parameters and latent trait distribution},
	journal = {Applied Psychological Measurement},
	year = {2013},
	volume = {37},
	number = {3},
	pages = {201--225},
	doi = {10.1177/0146621612470210}
}

@book{brennan_generalizability_2001,
	author = {Brennan, Robert L.},
	title = {Generalizability theory},
	year = {2001},
	publisher = {Springer},
	address = {New York}
}

@article{pustejovsky_meta-analysis_2022,
	author = {Pustejovsky, James E. and Tipton, Elizabeth},
	title = {Meta-analysis with robust variance estimation: {Expanding} the range of working models},
	journal = {Prevention Science},
	year = {2022},
	volume = {23},
	number = {3},
	pages = {425--438},
	doi = {10.1007/s11121-021-01246-3},
	issn = {1573-6695}
}

@article{tipton_small-sample_2015,
	author = {Tipton, Elizabeth and Pustejovsky, James E.},
	title = {Small-sample adjustments for tests of moderators and model fit using robust variance estimation in meta-regression},
	journal = {Journal of Educational and Behavioral Statistics},
	year = {2015},
	volume = {40},
	number = {6},
	pages = {604--634},
	doi = {10.3102/1076998615606099},
	issn = {1076-9986}
}

@article{holm_simple_1979,
	author = {Holm, Sture},
	title = {A simple sequentially rejective multiple test procedure},
	journal = {Scandinavian Journal of Statistics},
	year = {1979},
	volume = {6},
	number = {2},
	pages = {65--70},
	url = {https://www.jstor.org/stable/4615733},
	issn = {0303-6898}
}

@article{gilbert_item-level_2025,
	author = {Gilbert, Joshua B. and Himmelsbach, Zachary and Miratrix, Luke W. and Ho, Andrew D. and Domingue, Benjamin W.},
	title = {Item-level heterogeneity in value added models: {Implications} for reliability, cross-study comparability, and effect sizes},
	journal = {Journal of Educational and Behavioral Statistics},
	year = {2025},
	pages = {10769986251393339},
	doi = {10.3102/10769986251393339},
	issn = {1076-9986}
}

@article{cho_multilevel_2019,
	author = {Cho, Sun-Joo and Shen, Jianhong and Naveiras, Matthew},
	title = {Multilevel reliability measures of latent scores within an item response theory framework},
	journal = {Multivariate Behavioral Research},
	year = {2019},
	volume = {54},
	number = {6},
	pages = {856--881},
	doi = {10.1080/00273171.2019.1596780},
	issn = {0027-3171}
}

@article{hussey_aberrant_2025,
	author = {Hussey, Ian and Alsalti, Taym and Bosco, Frank and Elson, Malte and Arslan, Ruben},
	title = {An aberrant abundance of {Cronbach's} alpha values at .70},
	journal = {Advances in Methods and Practices in Psychological Science},
	year = {2025},
	volume = {8},
	number = {1},
	pages = {25152459241287123},
	doi = {10.1177/25152459241287123},
	issn = {2515-2459}
}

@article{spearman_correlation_1910,
	author = {Spearman, Charles},
	title = {Correlation calculated from faulty data},
	journal = {British Journal of Psychology},
	year = {1910},
	volume = {3},
	number = {3},
	pages = {271--295},
	doi = {10.1111/j.2044-8295.1910.tb00206.x}
}

@article{brown_some_1910,
	author = {Brown, William},
	title = {Some experimental results in the correlation of mental abilities},
	journal = {British Journal of Psychology},
	year = {1910},
	volume = {3},
	number = {3},
	pages = {296--322},
	doi = {10.1111/j.2044-8295.1910.tb00207.x}
}

@article{gelman_difference_2006,
	author = {Gelman, Andrew and Stern, Hal},
	title = {The difference between ``significant'' and ``not significant'' is not itself statistically significant},
	journal = {The American Statistician},
	year = {2006},
	volume = {60},
	number = {4},
	pages = {328--331},
	doi = {10.1198/000313006X152649}
}

@misc{zhang_realistic_2025,
	author = {Zhang, Lijin and Liu, Yiqing and Molenaar, Dylan and Domingue, Benjamin W.},
	title = {Realistic simulation of item difficulties},
	year = {2025},
	doi = {10.31234/osf.io/jbhxy_v1},
	url = {https://doi.org/10.31234/osf.io/jbhxy_v1}
}

@article{nalbandyan_signposts_2026,
	author = {Nalbandyan, Roza and Gilbert, Joshua B. and Franco, Vithor R. and Domingue, Benjamin W.},
	title = {Signposts on the path from nominal to ordinal scales: Moving from a discrete to a continuous view},
	journal = {Educational and Psychological Measurement},
	year = {2026},
	pages = {00131644261440556},
	publisher = {SAGE Publications Inc},
	doi = {10.1177/00131644261440556},
	url = {https://doi.org/10.1177/00131644261440556},
	issn = {0013-1644}
}

@article{domingue_intermodel_2024,
	author = {Domingue, Benjamin W. and Kanopka, Klint and Kapoor, Radhika and Pohl, Steffi and Chalmers, R. Philip and Rahal, Charles and Rhemtulla, Mijke},
	title = {The {InterModel} vigorish as a lens for understanding (and quantifying) the value of item response models for dichotomously coded items},
	journal = {Psychometrika},
	year = {2024},
	volume = {89},
	number = {3},
	pages = {1034--1054},
	doi = {10.1007/s11336-024-09977-2},
	url = {https://www.cambridge.org/core/journals/psychometrika/article/abs/intermodel-vigorish-as-a-lens-for-understanding-and-quantifying-the-value-of-item-response-models-for-dichotomously-coded-items/F61C75F6F945A5B13F73C6128EB83998},
	issn = {0033-3123, 1860-0980}
}

@article{gilbert_estimating_2025,
	author = {Gilbert, Joshua B. and Himmelsbach, Zachary and Soland, James and Joshi, Mridul and Domingue, Benjamin W.},
	title = {Estimating heterogeneous treatment effects with item-level outcome data: Insights from item response theory},
	journal = {Journal of Policy Analysis and Management},
	year = {2025},
	volume = {44},
	number = {4},
	pages = {1417--1449},
	doi = {10.1002/pam.70025},
	url = {https://onlinelibrary.wiley.com/doi/abs/10.1002/pam.70025},
	issn = {1520-6688}
}

@article{gilbert_conditional_2026,
	author = {Gilbert, Joshua B. and Young, William S. and Himmelsbach, Zachary and Ulitzsch, Esther and Domingue, Benjamin W.},
	title = {Conditional dependencies between response time and item discrimination: An item-level meta-analysis},
	journal = {Educational and Psychological Measurement},
	year = {2026},
	pages = {00131644261426972},
	publisher = {SAGE Publications Inc},
	doi = {10.1177/00131644261426972},
	url = {https://doi.org/10.1177/00131644261426972},
	issn = {0013-1644}
}

@misc{liu_comparing_2026,
	author = {Liu, Yiqing and Zhang, Lijin and Domingue, Benjamin},
	title = {Comparing compensatory and noncompensatory {MIRT} models using large-scale item response data},
	year = {2026},
	publisher = {OSF},
	doi = {10.31234/osf.io/r9b3y_v1},
	url = {https://doi.org/10.31234/osf.io/r9b3y_v1}
}

@article{greco_meta-analysis_2018,
	author = {Greco, Lindsey M. and O'Boyle, Ernest H. and Cockburn, Bethany S. and Yuan, Zhenyu},
	title = {Meta-analysis of coefficient alpha: {A} reliability generalization study},
	journal = {Journal of Management Studies},
	year = {2018},
	volume = {55},
	number = {4},
	pages = {583--618},
	doi = {10.1111/joms.12328},
	url = {https://onlinelibrary.wiley.com/doi/abs/10.1111/joms.12328},
	issn = {1467-6486}
}

@article{lance_sources_2006,
	author = {Lance, Charles E. and Butts, Marcus M. and Michels, Lawrence C.},
	title = {The sources of four commonly reported cutoff criteria: What did they really say?},
	journal = {Organizational Research Methods},
	year = {2006},
	volume = {9},
	number = {2},
	pages = {202--220},
	publisher = {SAGE Publications Inc},
	doi = {10.1177/1094428105284919},
	url = {https://journals.sagepub.com/doi/abs/10.1177/1094428105284919},
	issn = {1094-4281}
}

@article{kelley_confidence_2016,
	author = {Kelley, Ken and Pornprasertmanit, Sunthud},
	title = {Confidence intervals for population reliability coefficients: Evaluation of methods, recommendations, and software for composite measures},
	journal = {Psychological Methods},
	year = {2016},
	volume = {21},
	number = {1},
	pages = {69--92},
	publisher = {American Psychological Association},
	doi = {10.1037/a0040086},
	issn = {1939-1463},
	address = {US}
}

@article{sijtsma_part_2021,
	author = {Sijtsma, Klaas and Pfadt, Julius M.},
	title = {Part {II}: On the use, the misuse, and the very limited usefulness of {Cronbach’s} alpha: Discussing lower bounds and correlated errors},
	journal = {Psychometrika},
	year = {2021},
	volume = {86},
	number = {4},
	pages = {843--860},
	doi = {10.1007/s11336-021-09789-8},
	url = {https://www.cambridge.org/core/journals/psychometrika/article/part-ii-on-the-use-the-misuse-and-the-very-limited-usefulness-of-cronbachs-alpha-discussing-lower-bounds-and-correlated-errors/9B9076F685B9F4E50B4C03C679EDB23D},
	issn = {0033-3123, 1860-0980}
}

@article{berge_greatest_2004,
	author = {ten Berge, Jos M. F. and Sočan, Gregor},
	title = {The greatest lower bound to the reliability of a test and the hypothesis of unidimensionality},
	journal = {Psychometrika},
	year = {2004},
	volume = {69},
	number = {4},
	pages = {613--625},
	doi = {10.1007/BF02289858},
	url = {https://www.cambridge.org/core/journals/psychometrika/article/abs/greatest-lower-bound-to-the-reliability-of-a-test-and-the-hypothesis-of-unidimensionality/CAD6336326D42F5E0EE7D8D6EF875223},
	issn = {0033-3123, 1860-0980}
}

@article{cronbach_my_2004,
	author = {Cronbach, Lee J. and Shavelson, Richard J.},
	title = {My current thoughts on coefficient alpha and successor procedures},
	journal = {Educational and Psychological Measurement},
	year = {2004},
	volume = {64},
	number = {3},
	pages = {391--418},
	publisher = {SAGE Publications Inc},
	doi = {10.1177/0013164404266386},
	url = {https://doi.org/10.1177/0013164404266386},
	issn = {0013-1644}
}

@article{scherer_tutorial_2020,
	author = {Scherer, Ronny and Teo, Timothy},
	title = {A tutorial on the meta-analytic structural equation modeling of reliability coefficients},
	journal = {Psychological Methods},
	year = {2020},
	volume = {25},
	number = {6},
	pages = {747--775},
	publisher = {American Psychological Association},
	doi = {10.1037/met0000261},
	issn = {1939-1463},
	address = {US}
}

@article{lopez-ibanez_reliability_2024,
	author = {López-Ibáñez, Carmen and López-Nicolás, Rubén and Blázquez-Rincón, Desirée M. and Sánchez-Meca, Julio},
	title = {Reliability generalization meta-analysis: comparing different statistical methods},
	journal = {Current Psychology},
	year = {2024},
	volume = {43},
	number = {20},
	pages = {18275--18293},
	doi = {10.1007/s12144-023-05604-y},
	url = {https://link.springer.com/10.1007/s12144-023-05604-y},
	issn = {1046-1310, 1936-4733}
}

@article{weston_recommendations_2019,
	author = {Weston, Sara J. and Ritchie, Stuart J. and Rohrer, Julia M. and Przybylski, Andrew K.},
	title = {Recommendations for increasing the transparency of analysis of preexisting data sets},
	journal = {Advances in Methods and Practices in Psychological Science},
	year = {2019},
	volume = {2},
	number = {3},
	pages = {214--227},
	publisher = {SAGE Publications Inc},
	doi = {10.1177/2515245919848684},
	url = {https://doi.org/10.1177/2515245919848684},
	issn = {2515-2459}
}

@article{gilbert_sensitivity_2026,
  author = {Gilbert, Joshua B. and Soland, James G. and Domingue, Benjamin W.},
  title = {The sensitivity of value-added estimates to test scoring decisions},
  journal = {Educational Measurement: Issues and Practice},
  year = {2026},
  volume = {45},
  number = {1},
  pages = {e70011},
  doi = {10.1111/emip.70011}
}

@article{ahmed_heterogeneity_2025,
  author = {Ahmed, Ishita and Bertling, Masha and Zhang, Lijin and Ho, Andrew D. and Loyalka, Prashant and Xue, Hao and Rozelle, Scott and Domingue, Benjamin W.},
  title = {Heterogeneity of item-treatment interactions masks complexity and generalizability in randomized controlled trials},
  journal = {Journal of Research on Educational Effectiveness},
  year = {2025},
  volume = {18},
  number = {4},
  pages = {854--877},
  doi = {10.1080/19345747.2024.2361337}
}

@article{satterthwaite_approximate_1946,
  author = {Satterthwaite, F. E.},
  title = {An approximate distribution of estimates of variance components},
  journal = {Biometrics Bulletin},
  year = {1946},
  volume = {2},
  number = {6},
  pages = {110--114},
  doi = {10.2307/3002019}
}

@article{acciai_estimating_2023,
  author = {Acciai, Claudia and Schneider, Jesper W. and Nielsen, Mathias W.},
  title = {Estimating social bias in data sharing behaviours: An open science experiment},
  journal = {Scientific Data},
  year = {2023},
  volume = {10},
  number = {1},
  pages = {233},
  doi = {10.1038/s41597-023-02129-8}
}
